\documentclass[english]{article}
\usepackage{lmodern}

\usepackage[T1]{fontenc}
\usepackage[a4paper]{geometry}
\usepackage{dsfont}
\usepackage{color}
\usepackage[table,dvipsnames,svgnames,x11names,hyperref]{xcolor}
\usepackage{babel}
\usepackage{amsmath}
\usepackage{amsthm}
\usepackage{amssymb}
\usepackage{stmaryrd}
\usepackage{graphicx}
\usepackage{setspace}
\usepackage{esint}
\usepackage{caption}

\usepackage[authoryear]{natbib}
\usepackage[unicode=true,bookmarks=true,bookmarksnumbered=false,bookmarksopen=false,breaklinks=false,pdfborder={0 0 0},pdfborderstyle={},backref=page,colorlinks=true]{hyperref}
\hypersetup{pdftitle={VBPI-SIBranch},pdfauthor={Tianyu Xie, Frederick A. Matsen IV, Marc A. Suchard, Cheng Zhang},linkcolor=Blue,citecolor=Blue}
\usepackage{float}
\usepackage{subfig}
\usepackage{nicefrac}

\usepackage{epstopdf}
\usepackage{natbib}
\usepackage{url} 
\usepackage{caption}
\usepackage{multirow}
\usepackage{amssymb}
\usepackage{multicol}
\usepackage{booktabs}
\usepackage{tikz}
\usepackage{bbm}
\usepackage{bm}
\usetikzlibrary{positioning, bayesnet}
\usepackage{algorithm2e}
\RestyleAlgo{ruled}

\newtheorem{theorem}{Theorem}

\newtheorem{definition}{Definition}
\newtheorem{proposition}{Proposition}
\newtheorem{corollary}{Corollary}
\newtheorem{remark}{Remark}

\usepackage[textsize=tiny]{todonotes}

\makeatletter
\newcommand{\myfnsymbol}[1]{%
  \expandafter\@myfnsymbol\csname c@#1\endcsname
}

\newcommand{\@myfnsymbol}[1]{%
  \ifcase #1
  \or 1
  \or 2
  \or 3
  \or 4
  \or \TextOrMath{\textasteriskcentered}{*}
  \or \TextOrMath{\textdagger}{\dagger}
  \fi
}
\newcommand{\affiliationA}{\@myfnsymbol{1}}
\newcommand{\affiliationB}{\@myfnsymbol{2}}
\newcommand{\affiliationC}{\@myfnsymbol{3}}
\newcommand{\affiliationD}{\@myfnsymbol{4}}
\newcommand{\correspondingA}{\@myfnsymbol{5}}
\makeatother

\begin{document}

\title{ARTreeFormer: A Faster Attention-based Autoregressive Model for Phylogenetic Inference}

\author{
Tianyu Xie\textsuperscript{\affiliationA},
Yicong Mao\textsuperscript{\affiliationB},
Cheng Zhang\textsuperscript{\affiliationC,\correspondingA}
}

\date{
}

\renewcommand{\thefootnote}{\myfnsymbol{footnote}}
\maketitle
\footnotetext[1]{School of Mathematical Sciences, Peking University,
   Beijing, 100871, China. Email: tianyuxie@pku.edu.cn}%
\footnotetext[2]{School of Public Health, Peking University,
   Beijing, 100191, China. Email: ycmao@hsc.pku.edu.cn}%
\footnotetext[3]{School of Mathematical Sciences and Center for Statistical Science, Peking University, Beijing, 100871, China. Email: chengzhang@math.pku.edu.cn}
\footnotetext[5]{Corresponding author}%

\setcounter{footnote}{0}
\renewcommand{\thefootnote}{\fnsymbol{footnote}}

\begin{abstract}
Probabilistic modeling over the combinatorially large space of tree topologies remains a central challenge in phylogenetic inference.
Previous approaches often necessitate pre-sampled tree topologies, limiting their modeling capability to a subset of the entire tree space.
A recent advancement is ARTree, a deep autoregressive model that offers unrestricted distributions for tree topologies.
However, its reliance on repetitive tree traversals and inefficient local message passing for computing topological node representations may hamper the scalability to large datasets.
This paper proposes ARTreeFormer, a novel approach that harnesses fixed-point iteration and attention mechanisms to accelerate ARTree.
By introducing a fixed-point iteration algorithm for computing the topological node embeddings, ARTreeFormer allows fast vectorized computation, especially on CUDA devices. 
This, together with an attention-based global message passing scheme, significantly improves the computation speed of ARTree while maintaining great approximation performance.
We demonstrate the effectiveness and efficiency of our method on a benchmark of challenging real data phylogenetic inference problems.
\end{abstract}

\section{Introduction}

Unraveling the evolutionary relationships among species stands as a core problem in the field of computational biology.
This complex task, called \textit{phylogenetic inference}, is abstracted as the statistical inference on the hypothesis of shared history, i.e., \emph{phylogenetic trees}, based on collected molecular sequences (e.g., DNA, RNA) of the species of interest and a model of evolution. 
Phylogenetic inference finds its diverse applications ranging from genomic epidemiology \citep{Dudas2017-sb,Du_Plessis2021-tq,Attwood2022PhylogeneticAP} to the study of conservation genetics \citep{DeSalle2004-mm}. 
Classical approaches for phylogenetic inference includes maximum likelihood \citep{Felsenstein81}, maximum parsimony \citep{fitch1971parsimony}, and Bayesian approaches \citep{yang1997bayesian, Mau99, Larget1999MarkovCM}, etc.
Nevertheless, phylogenetic inference remains a hard challenge partially due to the combinatorially explosive size ($(2N-5)!!$ for unrooted bifurcating trees with $N$ species) of the phylogenetic tree topology space \citep{Whidden2014QuantifyingME, Dinh2017-oj}, which makes many common principles in phylogenetics, e.g., maximum likelihood and maximum parsimony, to be NP-hard problems \citep{Chor2005Maximumlikelihood, day1987complexity}.

Recently, the prosperous development of machine learning provides an effective and innovative approach to phylogenetic inference, and many efforts have been made for expressive probabilistic modeling of the tree topologies \citep{Hhna2012-pm, Larget2013-et, Zhang2018SBN, xie2023artree}.
A notable example among them is ARTree \citep{xie2023artree}, which provides a rich family of tree topology distributions and achieves state-of-the-art performance on benchmark data sets.
Given a specific order on the leaf nodes (also called the taxa order), ARTree generates a tree topology by sequentially adding a new leaf node to an edge of the current subtree topology at a time, according to an edge decision distribution modeled by graph neural networks (GNNs), until all the leaf nodes have been added.
Compared with previous methods such as conditional clade distribution (CCD) \citep{Larget2013-et} and subsplit Bayesian networks (SBNs) \citep{Zhang2018SBN}, an important advantage of ARTree is that it enjoys unconfined support over the entire tree topology space.
However, to compute the edge decision distribution in each leaf node addition step, ARTree requires sequential computations of topological node embeddings via tree traversals, which is hard to vectorize, making it prohibitive for phylogenetic inference for large numbers of species, as observed in \citet{xie2023artree}.
Besides, the message passing in ARTree only updates node features from their neighborhood, ignoring the important global information and would require multiple message passing rounds to obtain adequate information about trees.

To address the computational inefficiencies of ARTree, we propose ARTreeFormer, which enables faster ancestral sampling and probability evaluation by leveraging scalable system-solving algorithms and transformer architectures \citep{vaswani2017attention}.
More specifically, we replace the time-consuming tree traversal-based algorithm with a fixed-point iteration method for computing the topological node embeddings.
We also prove that, under a specific stopping criterion, the number of iterations required for convergence is independent of both the tree topology and the number of leaves.
To further reduce the computational cost, we introduce an attention-based global message passing scheme that captures tree-wide information in a single forward pass. Unlike ARTree, all components of ARTreeFormer can be fully vectorized across multiple tree topologies and nodes, allowing efficient batch-wise generation and evaluation. 
This makes ARTreeFormer particularly well-suited for large-batch training on CUDA\footnote{Compute Unified Device Architecture (CUDA) is a parallel computing platform and programming model developed by NVIDIA. It enables developers to use NVIDIA Graphics Processing Units (GPUs) for general-purpose processing (GPGPU), significantly accelerating computationally intensive tasks by leveraging the GPU's massive parallel processing capabilities.}-enabled devices, which is a standard setup in modern deep learning.
Our experiments demonstrate that ARTreeFormer achieves comparable or better performance than ARTree, while delivering approximately $10 \times$ faster generation and $6\times$ faster training on a benchmark suite covering maximum parsimony reconstruction, tree topology density estimation, and variational Bayesian phylogenetic inference tasks.

\section{Materials and methods}\label{sec:materials-and-methods}
In this section, we first introduce the necessary background, including the phylogenetic posterior, variational Bayesian phylogenetic inference, and the ARTree model for tree topology generation.
We then analyze the computational limitations of ARTree, which motivate the development of ARTreeFormer.
Finally, we present the two key components of ARTreeFormer: a fixed-point iteration method for computing topological node embeddings and an attention-based global message passing mechanism.

\subsection{Phylogenetic posterior}
The common structure for describing evolutionary history is a phylogenetic tree, which consists of a bifurcating tree topology $\tau$ and the associated non-negative branch lengths $\bm{q}$.
The tree topology $\tau$, which contains leaf nodes for the observed species and internal nodes for the unobserved ancestor species, represents the evolutionary relationship among these species.
A tree topology can be either rooted or unrooted.
In this paper, we only discuss unrooted tree topologies, but the proposed method can be easily adapted to rooted tree topologies.
The branch lengths $\bm{q}$ quantify the evolutionary intensity along the edges on $\tau$.
An edge is called a pendant edge if it connects one leaf node to an internal node.

Each leaf node on $\tau$ corresponds to a species with an observed biological sequence (e.g., DNA, RNA, protein).
Let $\bm{Y}=\{Y_1,\ldots,Y_M\}\in \Omega^{N\times M}$ be the observed sequences (with characters in $\Omega$) of $M$ sites over $N$ species.
A continuous-time Markov chain is commonly assumed to model the transition probabilities of the characters along the edges of a phylogenetic tree \citep{felsenstein2004inferring}. 
Under the assumption that different sites evolve independently and identically conditioned on the phylogenetic tree, the likelihood of observing sequences $\bm{Y}$ given a phylogenetic tree $(\tau,\bm{q})$ takes the form
\begin{equation}\label{eq:likelihood}
p(\bm{Y}|\tau,\bm{q}) =\prod_{i=1}^M \sum_{a^i}\eta(a^i_r)\prod_{(u,v)\in E}P_{a^i_u a^i_v}(q_{uv}),
\end{equation}
where $a^i$ ranges over all extensions of $Y_i$ to the internal nodes with $a^i_u$ being the character assignment of node $u$ ($r$ represents the root node), $E$ is the set of edges of $\tau$,
$q_{uv}$ is the branch length of the edge $(u,v)\in E$, $P_{jk}(q)$ is the transition probability from character $j$ to $k$ through an edge of length $q$, and $\eta$ is the stationary distribution of the Markov chain. 
Assuming a prior distribution $p(\tau,\bm{q})$ on phylogenetic trees, Bayesian phylogenetic inference then amounts to properly estimating the posterior distribution
\begin{equation}\label{eq:posterior}
p(\tau, \bm{q}|\bm{Y}) = \frac{p(\bm{Y}|\tau,\bm{q})p(\tau,\bm{q})}{p(\bm{Y})}\propto p(\bm{Y}|\tau,\bm{q})p(\tau,\bm{q}).
\end{equation}

\subsection{Variational Bayesian phylogenetic inference}\label{sec:vbpi}
By positing a phylogenetic variational family $Q_{\bm{\phi}, \bm{\psi}}(\tau,\bm{q})=Q_{\bm{\phi}}(\tau)Q_{\bm{\psi}}(\bm{q}|\tau)$ as the product of a tree topology model $Q_{\bm{\phi}}(\tau)$ and a conditional branch length model $Q_{\bm{\psi}}(\bm{q}|\tau)$, variational Bayesian phylogenetic inference (VBPI) converts the inference problem (\ref{eq:posterior}) into an optimization problem.
More specifically, VBPI seeks the best variational approximation by maximizing the following multi-sample lower bound
\begin{equation}\label{eq:lower-bound}
L^{K}(\bm{\phi},\bm{\psi}) = \mathbb{E}_{Q_{\bm{\phi},\bm{\psi}}(\tau^{1:K},\bm{q}^{1:K})} \log \left(\frac{1}{K}\sum_{i=1}^K\frac{p(\bm{Y}|\tau^i,\bm{q}^i) p(\tau^i, \bm{q}^i)}{Q_{\bm{\phi}}(\tau^i)Q_{\bm{\psi}}(\bm{q}^i|\tau^i)}\right),
\end{equation}
where $Q_{\bm{\phi},\bm{\psi}}(\tau^{1:K},\bm{q}^{1:K})=\prod_{i=1}^K Q_{\bm{\phi},\bm{\psi}}(\tau^{i},\bm{q}^{i})$.
In addition to the likelihood $p(\bm{Y},\tau,\bm{q})$ in the numerator of Eq~(\ref{eq:lower-bound}), one may also consider the parsimony score defined as the minimum number of character-state changes among all possible sequence assignments for internal nodes, i.e.,
\begin{equation}\label{eq:parsimony}
\mathcal{P}(\tau;\bm{Y}) = \sum_{i=1}^M \min_{a^i} \sum_{(u,v)\in E} \mathbb{I}(a_u^i\neq a_v^i),
\end{equation}
where the notations are the same as in Eq~(\ref{eq:likelihood}) \citep{Zhou2023PhyloGFN}.
The parsimony score $\mathcal{P}(\tau;\bm{Y})$ can be efficiently evaluated by the Fitch algorithm~\citep{fitch1971parsimony} in linear time.

The tree topology model $Q_{\bm{\phi}}(\tau)$ can take subsplit Bayesian networks (SBNs)~\citep{Zhang2018SBN} which rely on subsplit support estimation for parametrization, or ARTree \citep{xie2023artree} which is an autoregressive model using graph neural networks (GNNs) that provides distributions over the entire tree topology space.
A diagonal lognormal distribution is commonly used for the branch length model $Q_{\bm{\psi}}(\bm{q}|\tau)$ whose locations and scales are parameterized with heuristic features \citep{Zhang2019VBPI} or learnable topological features \citet{Zhang2023VBPIGNN}. 
More advanced models for branch lengths like normalizing flows \citep{Zhang2020VBPINF} or semi-implicit distributions \citep{xie2024vbpisibranch} are also applicable.
More details about VBPI can be found in Appendix \ref{app:vbpi}.

\subsection{ARTree for tree topology generation}
As an autoregressive model for tree topology generation, ARTree \citep{xie2023artree} decomposes a tree topology into a sequence of leaf node addition decisions and models the involved conditional probabilities with GNNs.
The corresponding tree topology generating process can be described as follows.
Let $\mathcal{X}=\{x_1,\ldots,x_N\}$ be the set of leaf nodes with a pre-defined order.
The generating procedure starts with a simple tree topology $\tau_3=(V_3,E_3)$ that has the first three nodes $\{x_1,x_2,x_3\}$ as the leaf nodes (which is unique), and keeps adding new leaf nodes according to the following rule.
Given an intermediate tree topology $\tau_n=(V_n,E_n)$ that has the first $n<N$ elements in $\mathcal{X}$ as the leaf nodes, i.e., an \emph{ordinal tree topology} of rank $n$ as defined in \citep{xie2023artree}, a probability vector $q_n \in \mathbb{R}^{|E_n|}$ over the edge set $E_n$ is first computed via GNNs.
Then, an edge $e_n\in E_n$ is sampled according to $q_n$ and the next leaf node $x_{n+1}$ is attached to it to form an ordinal tree topology $\tau_{n+1}$.
This procedure will continue until all the $N$ leaf nodes are added.
Although a pre-defined leaf node order is required, \citet{xie2023artree} shows that the performance of ARTree exhibits negligible dependency on this leaf node order.
Figure \ref{fig:artreeplot} is an illustration of ARTree.
See more details on ARTree in Appendix \ref{app:artree}.

\begin{figure}[h]
    \centering
    \includegraphics[width=\linewidth]{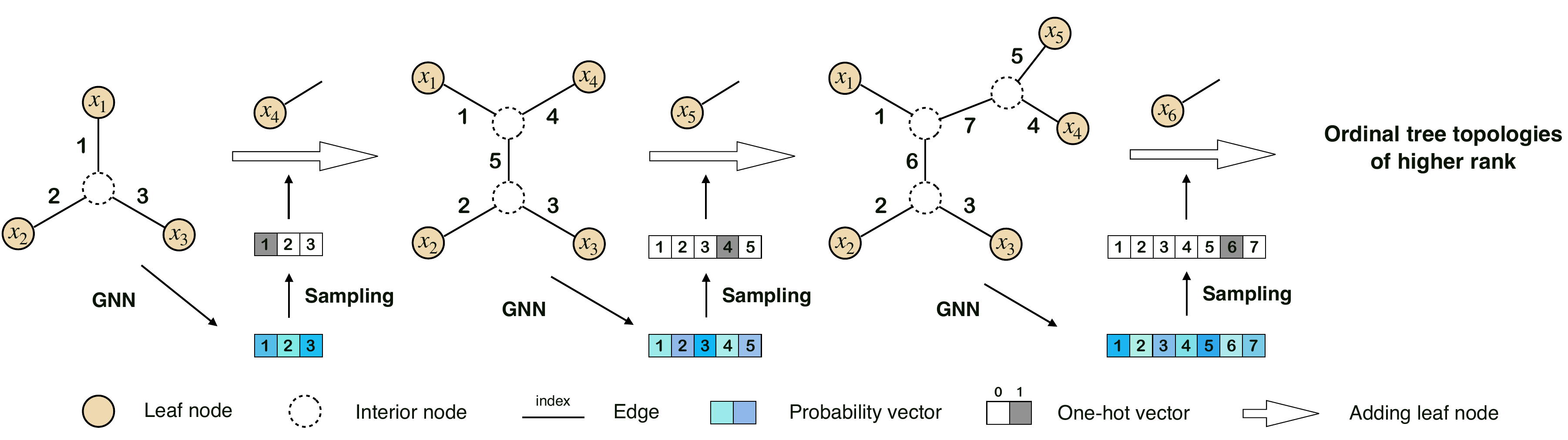}
    \caption{\textbf{An illustration of ARTree starting from the star-shaped tree topology with 3 leaf nodes.} This figure is from \citet{xie2023artree}.}
    \label{fig:artreeplot}
\end{figure}

Although ARTree enjoys unconfined support over the entire tree topology space and provides a more flexible family of variational distributions, it suffers from expensive computation costs (see Appendix E in \citet{xie2023artree}) which makes it prohibitive for phylogenetic inference when the number of species is large.
In the next two subsections, we discuss the computational cost of ARTree and then describe how it can be accelerated using fixed-point iteration and attention-based techniques. 

\subsection{Computational cost of ARTree}\label{sec:time-complexity}
In the $n$-th step of leaf node addition, ARTree includes the node embedding module and message passing module for computing the edge decision distribution, as detailed below.
Throughout this section, we use ``node embeddings'' (with dimension $N$) for the node information before message passing and ``node features'' (with dimension $d$) for those in and after message passing.

\paragraph{Node embedding module} 
The topological node embeddings $\{f_n(u)\in\mathbb{R}^N|u\in V_n\}$ of an ordinal tree topology $\tau_n=(V_n,E_n)$ in \citet{xie2023artree} are obtained by first assigning one-hot encodings to the leaf nodes and then minimizing the \textit{Dirichlet energy}
\begin{equation}\label{eq:global-dirichlet-energy}
\ell(f_n,\tau_n) := \sum_{(u,v)\in E_n}\left\|f_n(u)-f_n(v)\right\|^2,
\end{equation}
which is typically done by the two-pass algorithm \citep{Zhang2023VBPIGNN} (Algorithm \ref{alg:embedding} in Appendix \ref{app:artree}). 
This algorithm requires a traversal over a tree topology, which is hard to be efficiently vectorized across different nodes and different trees due to its sequential nature and the dependency on the specific tree topology shapes.
The complexity of computing the topological node embeddings is $O(Nn)$.
Finally, a multi-layer perceptron (MLP) is applied to all the node embeddings to obtain the node features with dimension $d$ enrolled in the computation of the following modules.

\paragraph{Message passing module}

At the beginning of message passing, assume the the initial node features are $\{f_n^0(u)\in\mathbb{R}^d|u\in V_n\}$, which are transformed from $\{f_n(u)\in \mathbb{R}^N|u\in V_n\}$ using MLPs with complexity $O(Nnd)$.
In the $l$-th round, these node features are updated by aggregating the information from their neighborhoods through
\begin{subequations}\label{eq:message-passing}
\arraycolsep=1.8pt
\def\arraystretch{1.5}
\begin{align}
m^l_n(u,v) &=F_{\textrm{message}}^l(f^l_n(u), f^l_n(v)),\label{mp-a}\\
f^{l+1}_n(v) &= F_{\textrm{updating}}^l\left(\{m^l_n(u,v);u\in \mathcal{N}(v)\}\right),\label{mp-b}
\end{align}
\end{subequations}
where the $l$-th message function $F_{\textrm{message}}^l$ and updating function $F_{\textrm{updating}}^l$ consist of MLPs.
After $L$ rounds of message passing, a recurrent neural network implemented by a gated recurrent unit (GRU) \citep{Gilmer2017NeuralMP} is then applied to help ARTree grasp the information from all previously generated tree topologies, i.e.,
\begin{equation}\label{eq:gru}
h_n(v) = \mathrm{GRU}(h_{n-1}(v), f^L_n(v)),
\end{equation}
where $h_n(v)$ is the hidden state of $v$. 
Eq~(\ref{eq:message-passing}) and (\ref{eq:gru}) are applied to the features of all the nodes on $\tau_n$ which require $O(Lnd^2)$ operations and is computationally inefficient especially when the number of leaf nodes is large.
Moreover, Eq~(\ref{eq:message-passing}) only updates the features of a node from its neighborhood, ignoring the global information of the full tree topology, and thus is called \textbf{local message passing} by us. We summarize the computational complexity of ARTree in Proposition \ref{prop:artree}.
\begin{proposition}[Time complexity of ARTree]
\label{prop:artree}
For generating $B$ tree topologies with $N$ leaf nodes, the time complexity of ARTree is $O(BN^3+BLN^2d^2+BN^3d)$. In the ideal case of perfect vectorization, the complexity of ARTree is $O(BN^2+LN)$.
\end{proposition}

Fig \ref{fig:artree-cpu-time} (left) demonstrates the run time of ARTree as the number of leaf nodes $N$ varies.
As $N$ increases, the total run time of ARTree grows rapidly and the node embedding module dominates the total time ($\approx 95\%$ on CUDA and $\approx 60\%$ on CPU),
which makes ARTree prohibitive when the number of leaf nodes is large.
The reason behind this is that compared to other modules, the node embedding module can not be easily vectorized w.r.t. different tree topologies and different nodes, resulting in great computational inefficiency.
It is worth noting that the computation time of the node embedding module on CUDA is even larger than that on CPU, which can be attributed to the inefficiency of CUDA for handling small tensors.

\begin{figure}[h]
\includegraphics[width=\linewidth]{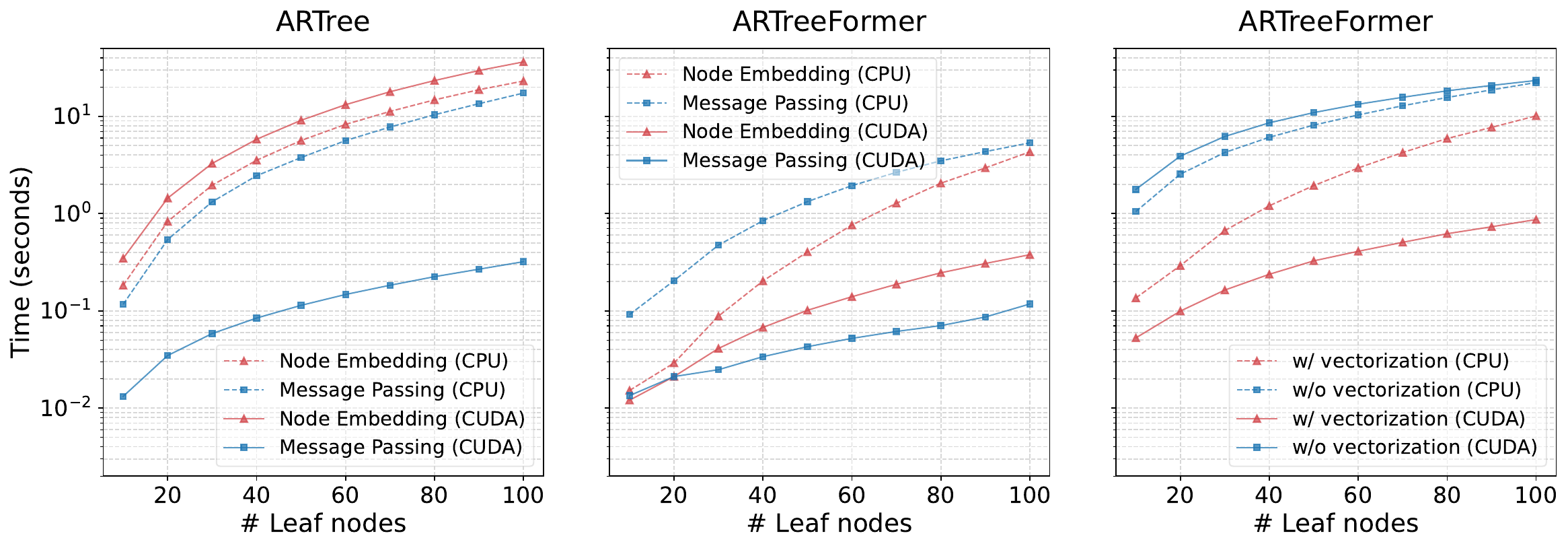}
\caption{\textbf{Time comparison between different models and devices}. Left \& Middle: Runtime of the node embedding module and message passing module for generating 128 tree topologies in a single batch using ARTree and ARTreeFormer. Right: The runtime of ARTreeFormer for generating 128 tree topologies with or without vectorization across batched tree topologies.
CPU means running on a cluster of 16 2.4GHz CPUs, and CUDA means running on a single NVIDIA A100 GPU.
All these results are averaged over 10 independent trials.
}
\label{fig:artree-cpu-time}
\end{figure}

\subsection{Accelerated computation of edge decision distributions}\label{sec:artreeformer}
In this subsection, we propose ARTreeFormer, which introduces a fast fixed-point iteration algorithm for topological node embeddings and an attention-based global message passing scheme to accelerate the training and sampling in ARTree.
In what follows, we present our approach for modeling the edge decision distribution at the $n$-th step.

\paragraph{Fixed-point iteration for topological node embedding}
Instead of solving the minimization problem of $\ell(f_n,\tau_n)$ in Eq~(\ref{eq:global-dirichlet-energy}) with the time-consuming two-pass algorithm, we reformulate it as a fixed-point iteration algorithm.
For a tree topology $\tau_n=(V_n,E_n)$, denote the set of leaf nodes by $\mathcal{X}_n$, the set of internal branches by $V_n^o$, and the set of internal nodes by $E_n^o=\{(u,v)| u,v\in V_n^{o}\}$.
Note that the global minimum of $\ell(f_n,\tau_n)$ satisfies
\begin{equation}
\label{eq:global-minimum}
\left\{
\begin{array}{lc}
f_n(u)= \frac{1}{3}\sum_{v\in \mathcal{N}(u)}f_n(v) , & u \in V_n^o;\\  f_n(x_i) = \delta_{i},& x_i \in \mathcal{X}_n; 
\end{array}
\right.
\end{equation}
where $\mathcal{N}(u)$ is the set of neighbors of $u$ and $\delta_{i}$ is a one-hot vector of length $n$ with the $1$ at the $i$-th position.
Let $\bar{\mathcal{F}}_n=\{f_n(u)\in\mathbb{R}^n|u\in V_n\}\in \mathbb{R}^{(2n-2)\times n}$ and $\mathcal{F}_n=\{f_n(u)\in\mathbb{R}^n|u\in V_n^o\}\in \mathbb{R}^{(n-2)\times n}$, then $\bar{\mathcal{F}}_n=(I_n, \mathcal{F}'_n)'$.
Consider a matrix $\Bar{A}_n$ satisfying
\begin{equation}\label{eq:matrix-decomp}
\bar{A}_n = \left(
\begin{array}{cc}
I_n& 0_{n-2} \\
C_n/3& A_n/3 
\end{array}
\right),
\end{equation}
where $C_n({i,j})=\mathbb{I}_{(u_i, x_j)\in E_n}$ and $A_n({i,j})=\mathbb{I}_{(u_i,u_j)\in E_n}$ ($u_i$ denotes the $i$-th node, leaf nodes are indexed as the first $n$ nodes).
Note that $A_n$ is exactly the adjacency matrix of $(V_n^o, E_n^o):=\tau^o_n$.
We call $A_n$ the \emph{interior adjacency matrix} and $C_n$ the \emph{leaf-interior cross adjacency matrix} of $\tau_n$.
The system (\ref{eq:global-minimum}) is then equivalent to $\bar{A}_n\bar{F}_n=\bar{F}_n$, i.e.,
\begin{equation}\label{eq:reduced-equation}
\mathcal{F}_n = \frac{A_n}{3}\mathcal{F}_n + \frac{C_n}{3}.
\end{equation}
This inspires the following fixed-point iteration algorithm:
\begin{equation}\label{eq:fixed-point-iteration}
\mathcal{F}_n^{(m+1)} = \frac{A_n}{3}\mathcal{F}_n^{(m)} + \frac{C_n}{3};\ \mathcal{F}_n^{(0)}=\mathcal{F}_n^{(0)}.
\end{equation}
In practice, we set all the entries to $\mathcal{F}_n^{(0)}$ as $1/n$.
Finally, after obtaining the solution $\mathcal{F}_n^{\ast}$, we pad $N-n$ zeros on its right so that the resulting length-$N$ node embeddings can be fed into the message passing module.
Theorem \ref{thm:spectrum} and Corollary \ref{thm:convergence} prove that the fixed-point iteration (\ref{eq:fixed-point-iteration}) will converge to the unique solution of Eq~(\ref{eq:reduced-equation}) with a uniform speed for all tree topologies $\tau_n$, the number of leaves $n$, and the initial condition $\mathcal{F}_n^{(0)}$.

\begin{theorem}\label{thm:spectrum}
For a tree topology $\tau_n$ with $n$ leaf nodes, let $\tau^o_n$ be the subgraph of $\tau_n$ which only contains the internal nodes and the branches among them and $A_n$ be the interior adjacency matrix of $\tau_n$.
Let $\rho(\tau^o_n)$ be the spectral radius of $\tau^o_n$ defined as $\rho(\tau^o_n) = \lambda_{\max}(A_n)$, where $\lambda_{\max}(\cdot)$ denotes the largest absolute eigenvalue of a matrix.
Then for any $\tau_n$ and $n$, it holds
\begin{equation*}
\rho(\tau^o_n) \leq 2\sqrt{2}.
\end{equation*}
\end{theorem}
\begin{proof}\
Without loss of generality, we select a node $c$ in $\tau^o_n$ as the ``root node'' and denote the distance between a node $u$ and $c$ by $d_u$, which induces a hierarchical structure on $\tau^o_n$.
Consider a matrix $D = \mathrm{diag}\{2^{d_u/2}, u\in V^o\}$ and it holds that $DA_nD^{-1}$ and $A_n$ share the same eigenvalues.
Note that each row of $A$ and $DA_nD^{-1}$ has up to 3 non-zero entries.
Now for each row of $DA_nD^{-1}$ and the corresponding node $u$, we make the following analysis.
\begin{itemize}
    \item If $u=c$, then each non-zero entry in this row equals to $1/\sqrt{2}$.
    \item If $u$ is a leaf node, then there is only one non-zero entry $\sqrt{2}$ in this row.
    \item For the remaining cases, as $u$ have one parent node and at most two child nodes, there is a $\sqrt{2}$ and at most two $1/\sqrt{2}$ entries in this row.
\end{itemize}

For all these cases, the row sum is less than or equal to $2\sqrt{2}$ which consistently holds for arbitrary topological structures of $\tau^o_n$ and the number of nodes $n$.
By the Perron–Frobenius theorem for positive matrices, $\lambda_{\max}(DA_nD^{-1})$ is upper bounded by the largest row sum of $DA_nD^{-1}$. Therefore, we conclude that $\rho(\tau^o_n) \leq 2\sqrt{2}$.
This proof is inspired by \citet[Section 4.2]{spielman2025graph}.
\end{proof}
\begin{corollary}\label{thm:convergence}
The fixed-point iteration algorithm (\ref{eq:fixed-point-iteration}) will converge linearly with rate $\frac{2\sqrt{2}}{3}$.
\end{corollary}
\begin{proof}\
Let $\mathcal{F}_n^\ast$ be the solution to $\mathcal{F}_n^\ast = A_n\mathcal{F}_n^\ast/3+C_n/3$. The existence and uniqueness of $\mathcal{F}_n^\ast$ are guaranteed by the fact that $I-A_n/3$ is a full-rank matrix.
Subtracting $\mathcal{F}_n^\ast$ from both sides lead to
\[
\mathcal{F}_n^{(m+1)} - \mathcal{F}_n^\ast = (A_n/3)(\mathcal{F}_n^{(m)} - \mathcal{F}_n^\ast)
\]
and thus
\[
\|\mathcal{F}_n^{(m+1)} - \mathcal{F}_n^\ast\|_2 \leq (\|A\|_2/3)\|\mathcal{F}_n^{(m)} - \mathcal{F}_n^\ast\|_2
\]
By Theorem \ref{thm:spectrum}, we conclude that $\|\mathcal{F}_n^{(m)} - \mathcal{F}_n^\ast\|_2\leq\left(\frac{2\sqrt{2}}{3}\right)^m \|\mathcal{F}_n^{(0)} - \mathcal{F}_n^\ast\|_2$.
\end{proof}

Unlike the two-pass algorithm for Dirichlet energy minimization, the fixed-point iteration can be easily vectorized over different tree topologies and nodes, making it suitable for fast computation on CUDA.
By using $\|\mathcal{F}_n^{(m)} - \mathcal{F}_n^\ast\|_2/n < \varepsilon$ as the stopping criterion, the required number of iterations $M_\varepsilon$ is a constant independent of the tree topologies.
Moreover, by noting that the fixed-point iteration (\ref{eq:fixed-point-iteration}) is equivalent to $\bar{\mathcal{F}}_n^{(2^{m+1})} = \bar{A}_n^{2^m} \bar{\mathcal{F}}_n^{(2^m)}$ which repetitively updates $\bar{A}_n^{2^{m+1}} = \left(\bar{A}_n^{2^{m}}\right)^2$, the number of iterations can be further reduced to $\log_2 M_{\varepsilon}$ and we call this strategy \emph{the power trick}.
With the power trick, the computational complexity of fixed-point iteration over $B$ tree topologies can be reduced to $O(Bn^2\log_2 M_\varepsilon)$.

\begin{remark}
For the complexity estimation $O(Bn^2\log_2 M_\varepsilon)$, the computation over the dimension $B$ and $n$ can be efficiently vectorized, while the computation over $\log M_\varepsilon$ is still sequential.
\end{remark}

\begin{remark}
After adding a new leaf node to $\tau_n$, a local modification can be applied to $A_n$ and $C_n$ to form the $A_{n+1}$ and $C_{n+1}$. Therefore, the time complexity of computing the adjacency matrices of a tree topology is $O(1)$.
Algorithm \ref{alg:fixed-point} shows the full procedure of the fixed-point iteration when autoregressively building the tree topology.
\end{remark}

\begin{algorithm}[h]
\caption{Fixed-point Algorithm for Topological Node Embeddings}
\label{alg:fixed-point}
\KwIn{A decision sequence $D=(e_3,\ldots,e_{N-1})$ corresponding to $\tau$. A threshold value $\varepsilon$.}
\KwOut{The topological node embeddings of each subtree $\tau_n$}
Initialize the adjacency matrices $A_3$ and $C_3$\;
\For{
$n = 3, \ldots, N-1$
}{
Compute $\bar{A}_n$ from $A_n$ and $C_n$ using Eq~(\ref{eq:matrix-decomp})\;
Initialize $\mathcal{F}_n^{(1)}=(I_n, 1_n/n)$ and $m=1$\;
Compute $\mathcal{F}_n^{(2)}=\bar{A}\mathcal{F}_n^{(1)}$ and $\bar{A}^2=\bar{A}\cdot \bar{A}$\;
\While{$\|\mathcal{F}_n^{(2^m)}-\mathcal{F}_n^{(2^{m-1})}\|_2\geq \varepsilon$}{
Compute $\mathcal{F}_n^{(2^{m+1})} = \bar{A}_n^{2^{m}}\mathcal{F}_n^{(2^{m})}$\;
Compute $\bar{A}_n^{2^{m+1}}=\bar{A}_n^{2^{m}}\cdot \bar{A}_n^{2^{m}}$;
$m=m+1$\;
}
Pad $N-n$ zeros on the right of each row of  $\mathcal{F}^{(2^m)}_n$\;
Add the new leaf node $x_{n+1}$ to the edge $e_n$\;
Locally modify the adjacency matrices $A_n$ and $C_n$ to obtain $A_{n+1}$ and $C_{n+1}$.
}
\end{algorithm}

\paragraph{Attention-based global message passing}
After obtaining the topological node embeddings $\mathcal{F}_n^\ast$ with the fixed-point iteration algorithm, it is fed into a message passing module to form the distribution over edge decisions.
To design an edge distribution that captures the global information of the tree topology, we substitute the GNNs with the powerful attention mechanism \citep{vaswani2017attention}.
Specifically, we first use the attention mechanism to compute a graph representation vector $r_n\in\mathbb{R}^d$, i.e.,
\begin{subequations}\label{eq:global-representation}
\arraycolsep=1.8pt
\def\arraystretch{1.5}
\begin{align}
\bar{r}_n &=F_{\textrm{graph}}(q_n, L(\mathcal{F}_n^\ast), L(\mathcal{F}_n^\ast)),\label{eq:global-representation-a}\\
r_n &= R_{\textrm{graph}}(\bar{r}_n),\label{eq:global-representation-b}
\end{align}
\end{subequations}
where $F_{\textrm{graph}}$ is the graph pooling function implemented as a multi-head attention block \citep{vaswani2017attention}, $R_{\textrm{graph}}$ is the graph readout function implemented as a 2-layer MLP, $q_n\in\mathbb{R}^d$ is a learnable query vector, and $L:\mathbb{R}^N\to \mathbb{R}^d$ is an embedding map implemented as a 2-layer MLP.
Here, the multi-head attention block $M=\mathrm{MHA}(Q,K,V)$ is defined as
\begin{subequations}\label{eq:mha}
\arraycolsep=1.8pt
\def\arraystretch{1.5}
\begin{align}
H_i&=\mathrm{softmax}\left(\frac{(QW_i^Q)(KW_i^K)'}{\sqrt{d/h}}\right)\cdot (VW_i^V),\\
M &= \mathrm{CONCAT}\left(H_1,\ldots,H_h\right)W^O,\label{eq:mha-b}
\end{align}
\end{subequations}
where $W_i^Q,W_i^K,W_i^V\in\mathbb{R}^{d\times\frac{d}{h}}$ and $W^O\in\mathbb{R}^{d\times d}$ are learnable matrices, $h$ is the number of heads, and $\mathrm{CONCAT}$ is the concatenation operator along the node feature axis.
Intuitively, we have used a global vector $q_n$ to query all the node features and obtained a representation vector $r_n$ for the whole tree topology $\tau_n$.
We emphasize that Eq~(\ref{eq:global-representation}) enjoys time complexity $O(nd+d^2)$ instead of the $O(n^2d+nd^2)$ of common multi-head attention blocks, as $q_n$ is a one-dimensional vector.

We now compute the edge decision distribution to decide where to add the next leaf node, similarly to ARTree.
To incorporate global information into the edge decision, we utilize the global representation vector $r_n$ to compute the edge features.
Concretely, the feature of an edge $e=(u,v)$ is formed by
\begin{subequations}\label{eq:edge-readout}
\arraycolsep=1.8pt
\def\arraystretch{1.5}
\begin{align}
p_n(e) &= F_{\textrm{edge}}\left(\{f_n(u),f_n(v)\}\right),\\
r_n(e) &= R_{\textrm{edge}}\left(\mathrm{CONCAT}(p_n(e), r_n)+b_n\right),\label{eq:edge-readout-b}
\end{align}
\end{subequations}
where $F_{\textrm{edge}}$ is an invariant edge pooling function implemented as an elementwise maximum operator, $R_{\textrm{edge}}$ is the edge readout function implemented as a 2-layer MLP with scalar output, and $b_n$ is the sinusoidal positional embedding \citep{vaswani2017attention} of the time step $n$.
The time complexity of these MLPs in Eq~(\ref{eq:edge-readout}) is $O(nd^2)$.

\paragraph{Edge decision distribution}
Similarly to ARTree, we builds the edge decision distributions in ARTreeFormer in an autoregressive way.
That is, we directly readout the representation vector $r_n$
to calculate the edge decision distribution $Q_{\bm{\phi}}(\cdot|e_{<n})$ using
\begin{equation}\label{eq:probability-vector}
Q_{\bm{\phi}}(\cdot|e_{<n})=\mathrm{Discrete}(\alpha_n),\ \alpha_n = \mathrm{softmax}\left([r_n(e)]_{e\in E_n}\right),
\end{equation}
and grow $\tau_n$ to $\tau_{n+1}$ by attaching the next leaf node $x_{n+1}$ to the sampled edge (Algorithm \ref{alg:generation}).

The above node embedding module and message passing module circularly continue until an ordinal tree topology of $N$, $\tau_N$, is constructed, whose ARTreeFormer-based probability is defined as
\begin{equation}
Q_{\bm{\phi}}(\tau_N) = \prod_{n=3}^{N-1}Q_{\bm{\phi}}(e_n|e_{<n}),
\end{equation}
where $\bm{\phi}$ are the learnable parameters and $Q_{\bm{\phi}}(e_n|e_{<n})$ is defined in Eq~(\ref{eq:probability-vector}).
We summarize the time complexity of ARTreeFormer in Proposition \ref{prop:artreeformer}.
\begin{proposition}[Time complexity of ARTreeFormer]
\label{prop:artreeformer}
For generating $B$ tree topologies with $N$ leaf nodes, the time complexity of ARTreeFormer is $O(BN^3\log M_\varepsilon+BN^2d^2)$.
But in the ideal case of perfect vectorization, the time complexity of ARTreeFormer is $O(N(\log M_\varepsilon+1))$, where $\log M_\varepsilon$ is a constant independent of $N$.
\end{proposition}

Compared to ARTree, the greatly improved computational efficiency of ARTreeFormer mainly comes from two aspects. 
\textbf{First}, the fixed-point iteration algorithm in ARTreeFormer for topological node embeddings can be easily vectorized across different tree topologies and different nodes, since they do not rely on traversals over tree topologies.
\textbf{Second}, the global message passing in ARTreeFormer forms the global representation only in one pass through the attention mechanism instead of gathering the neighborhood information repetitively with GNNs.
We depict the the pipeline of the leaf node addition of ARTreeFormer in Fig \ref{fig:artreeformer}.

\begin{figure}[h]
    \centering
    \includegraphics[width=\linewidth]{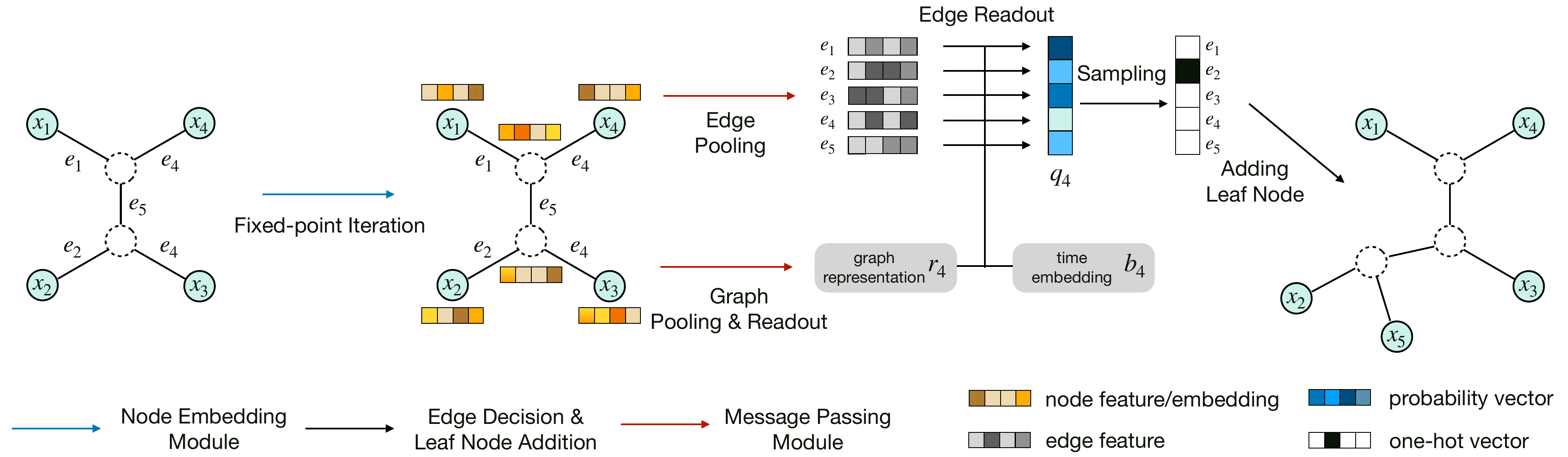}
    \vspace{0.2cm}
    \caption{\bf An illustration of ARTreeFormer for growing an ordinal tree topology $\tau_4$ of rank $4$ to an ordinal tree topology $\tau_5$ of rank $5$.
    }
    \label{fig:artreeformer}
\end{figure}

In Fig \ref{fig:artree-cpu-time} (left, middle), for the node embedding module on CPU/CUDA, the time consumption of ARTreeFormer is less than $10\%$ of ARTree, and this number is $50\%$ for the message passing module on CPU/CUDA. 
Moreover, both two modules of ARTreeFormer enjoy a significant time consumption drop on CUDA compared to CPU, since CUDA is more powerful at handling large tensor multiplications.
To further verify the vectorization capability of ARTreeFormer, we compare the runtime for generating tree topologies with or without vectorization (i.e., simultaneously or sequentially) in Fig \ref{fig:artree-cpu-time} (right), where vectorization greatly improves computational efficiency.
The vectorization capability of ARTreeFormer further allows for training with a larger batch size (note the batch size is 10 in ARTree), which is a common setting in modern deep learning methods.

\section{Results}
In this section, we demonstrate the effectiveness and efficiency of ARTreeFormer on three benchmark tasks: maximum parsimony, tree topology density estimation (TDE), and variational Bayesian phylogenetic inference (VBPI).
Although the pre-selected leaf node order in ARTreeFormer may not be related to the relationships among species, this evolutionary information is already contained in the training data set (for TDE) or the target posterior distribution (for maximum parsimony and VBPI), and thus can be learned by ARTreeFormer.
Noting that the main contribution of ARTreeFormer is improving the tree topology model, we select the first two tasks because they only learn the tree topology distribution and can better demonstrate the superiority of ARTreeFormer.
The third task, VBPI, is selected as a standard benchmark task for Bayesian phylogenetic inference and evaluates how well ARTreeFormer collaborates with a branch length model.
It should be emphasized that we mainly pay attention to the computational efficiency improvement of ARTreeFormer and only expect it to attain similar accuracy to ARTree.
Throughout this section, the run times of ARTree are reproduced using its official codebase\footnote{\href{https://github.com/tyuxie/ARTree}{\texttt{https://github.com/tyuxie/ARTree}}}. 

\paragraph{Experimental setup} 
For TDE and VBPI, we perform experiments on eight data sets which we will call DS1-8. 
These data sets, consisting of sequences from 27 to 64 eukaryote species with 378 to 2520 site observations, are commonly used to benchmark phylogenetic MCMC methods \citep{Hedges90, Garey96, Yang03, Henk03, Lakner08, Zhang01, Yoder04, Rossman01, Hhna2012-pm, Larget2013-et, Whidden2014QuantifyingME}.
For the Bayesian setting in MrBayes runs \citep{ronquist2012mrbayes}, we assume a uniform prior on the tree topologies, an i.i.d. exponential prior $\mathrm{Exp}(10)$ on branch lengths, and the simple Jukes \& Cantor (JC) substitution model \citep{jukes1969evolution}.
We use the same ARTreeFormer structure across all the data sets for all three experiments.
Specifically, we set the dimension of node features to $d=100$, following \citet{xie2023artree}. 
The number of heads in all the multi-head attention blocks is set to $h=4$.
All the activation functions for MLPs are exponential linear units (ELUs)~\citep{ELU}.
We add a layer normalization block after each linear layer in MLPs and before each multi-head attention block, which stabilizes training and reduces its sensitivity to optimization tricks \citep{xiong2020transformerlayernorm}.
We also add a residual block after the multi-head attention block in the message passing step, which is standard in transformers.
The taxa order is set to the lexicographical order of the corresponding species names.
All models are implemented in PyTorch \citep{Paszke2019PyTorchAI} and optimized with the Adam \citep{ADAM} optimizer. 
All the experiments are run and all the runtimes are measured on a single CUDA-enabled NVIDIA A100 GPU.
The learning rate for ARTreeFormer is set to 0.0001 in all the experiments, which is the same as in ARTree \citep{xie2023artree}.

\subsection{Maximum parsimony problem}
We first test the performance of ARTreeFormer on solving the maximum parsimony problem.
We reformulate this problem as a Bayesian inference task with the target distribution
$P(\tau) = \exp(-\mathcal{P}(\tau,\bm{Y}))/Z$,
where $\mathcal{P}(\tau,\bm{Y})$ is the parsimony score defined in Eq~(\ref{eq:parsimony}) and $Z=\sum_{\tau}\exp(-\mathcal{P}(\tau,\bm{Y}))$ is the normalizing constant.
To fit a variational distribution $Q_{\bm{\phi}}(\tau)$, we maximize the following (annealed) multi-sample lower bound ($K=10$) in the $t$-th iteration
\begin{equation}
\mathcal{L}(\bm{\phi};\beta_t) =\mathbb{E}_{Q_{\bm{\phi}}(\tau^{1:K})} \log\left(\frac{1}{K}\sum_{i=1}^{K}\frac{\exp\left(-\beta_t\mathcal{P}(\tau_i,\bm{Y})\right)}{Q_{\bm{\phi}}(\tau_i)}\right),
\end{equation}
where $Q_{\bm{\phi}}(\tau^{1:K})=\prod_{i=1}^KQ_{\bm{\phi}}(\tau^{i})$ and  $\beta_t$ is the annealing schedule.
We set $\beta_t=\min\{1, 0.001+t/200000\}$ and collect the results after 400000 parameter updates.
We use the VIMCO estimator \citep{Mnih2016vimco} to estimate the stochastic gradients of $\mathcal{L}(\bm{\phi})$.

Fig \ref{fig:parsimony} shows the performances of different methods for the maximum parsimony problem on DS1. 
We run the state-of-the-art parsimony analysis software PAUP$^\ast$ \citep{swofford2003paup} to form the ground truth, which contains tree topologies with parsimony scores ranging from 4040 to the optimal score 4026. 
The left plot of Fig \ref{fig:parsimony} shows that both ARTreeFormer and ARTree can identify the most parsimonious tree topology found by PAUP$^\ast$ and provide comparably accurate posterior estimates.
In the right plot of Fig \ref{fig:parsimony}, the horizontal gap between two curves reflects the ratio of times needed to reach the same lower bound or negative parsimony score.
We see that ARTreeFormer is around four times faster than ARTree.

\begin{figure}[h]
    \centering
    \includegraphics[width=0.7\linewidth]{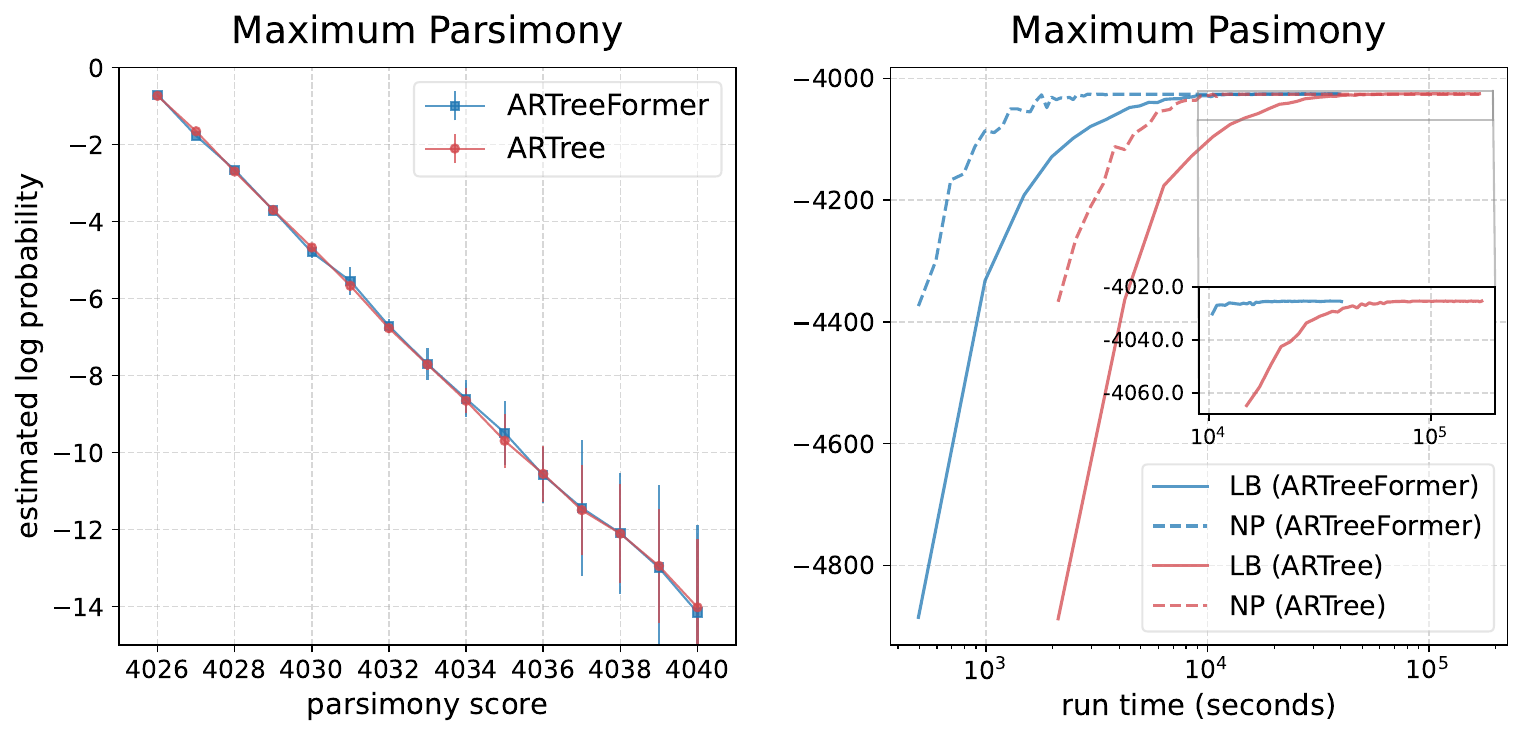}
    \caption{\textbf{Performances of ARTree and ARTreeFormer on the maximum parsimony problem.}
    Left: The estimated log probability $\log Q(\tau)$ versus the parsimony score $\mathcal{P}(\tau,\bm{Y})$ on DS1. 
    For different tree topologies with the same parsimony score, the mean of the estimated log probabilities is plotted as a dot with the standard deviation as the error bar.
    Right: The 10-sample lower bound (LB) and the negative parsimony score (NP) as a function of the run time on DS1.
    }
    \vspace{-0.3cm}
    \label{fig:parsimony}
\end{figure}

\subsection{Tree topology density estimation}
We further investigate the ability of ARTreeFormer to model tree topologies on the TDE task.
To construct the training data set, we run MrBayes \citep{ronquist2012mrbayes} on each data set with 10 replicates of 4 chains and 8 runs until the runs have ASDSF (the standard convergence criteria used in MrBayes) less than 0.01 or a maximum of 100 million iterations, collect the samples every 100 iterations, and discard the first 25\%, following \citet{Zhang2018SBN}.
The ground truth distributions are obtained from 10 extremely long single-chain MrBayes runs, each for one billion iterations,
where the samples are collected every 1000 iterations, with the first 25\% discarded as burn-in.
We train ARTreeFormer via maximum likelihood estimation using stochastic gradient ascent.
We compare ARTreeFormer to ARTree and SBN baselines: 
(i) for SBN-EM and SBN-EM-$\alpha$, the SBN model is optimized using the expectation-maximization (EM) algorithm, as done in \citet{Zhang2018SBN}; 
(ii) for SBN-SGA and ARTree, the corresponding models are fitted via stochastic gradient ascent, similarly to ARTreeFormer. For SBN-SGA, ARTree, and ARTreeFormer, the results are collected after 200000 parameter updates with a batch size of 10.

The left plot in Fig \ref{fig:DS1-TDE} shows a significant reduction in the training time and evaluation time of ARTreeFormer compared to ARTree on DS1-8.
To further demonstrate the benefit of vectorization over different tree topologies, we train ARTreeFormer on DS1 with different batch sizes, and report the Kullback-Leibler (KL) divergences in Fig \ref{fig:DS1-TDE} (right).
We see that a large batch size will only lead to a minor training speed drop, but will significantly benefit the training accuracy.
We can also observe a saturated approximation accuracy with a sufficiently large batch size.

\begin{figure}[h]
    \centering
    \includegraphics[width=0.7\linewidth]{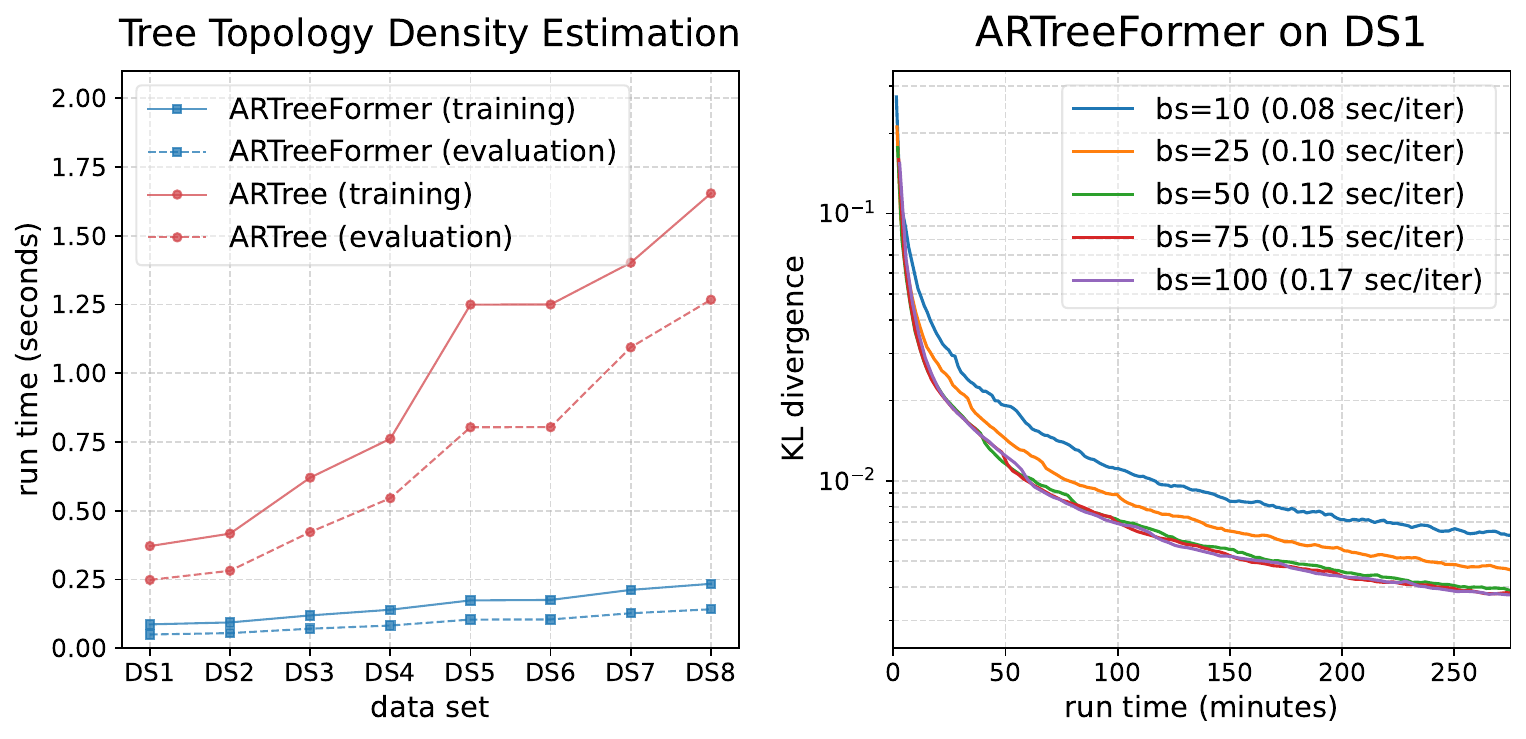}
    \caption{\textbf{Performance of ARTree and ARTreeFormer on the TDE task.}
    Left: The training time (per iteration) and evaluation time (per evaluating the probabilities of 10 tree topologies) of ARTree and ARTreeFormer across eight benchmark data sets for TDE (averaged over 100 trials).
    Right: The KL divergence to the ground truth on DS1 obtained by ARTreeFormer, as the batch size (bs) varies. The training speed measured by seconds per iteration is reported in the parenthesis.
    }
    \label{fig:DS1-TDE}
\end{figure}

The KL divergences between the ground truth and the probability estimation are reported in Table \ref{tab:tde}.
Although ARTreeFormer has only one attention layer for node features, it performs on par or better than ARTree, and consistently outperforms the SBN-related baselines, across all data sets.
See the probability estimation on individual tree topologies and an ablation study about the hyperparameters in Appendix \ref{app:tde-results}.

\begin{table}[h]
    \caption{\textbf{KL divergences to the ground truth of different methods across eight benchmark data sets.}
    }
    \label{tab:tde}
    \centering
    \resizebox{\linewidth}{!}{
    \begin{tabular}{cccccccccc}
    \toprule
    \multirow{2}{*}{Data set} &\multirow{2}{*}{\# Taxa} & \multirow{2}{*}{\# Sites}  &\multirow{2}{*}{\shortstack{Sampled trees}}&\multirow{2}{*}{\shortstack{GT trees}} &\multicolumn{5}{c}{KL divergence to ground truth} \\
\cmidrule(l){6-10}
       &&&& & SBN-EM & SBN-EM-$\alpha$ & SBN-SGA &  ARTree & ARTreeFormer\\
       \midrule
       DS1 & 27 & 1949 & 1228&2784  & 0.0136 & 0.0130 & 0.0504  & \textbf{0.0045} & 0.0065 \\
       DS2 & 29 & 2520 & 7 &42 & 0.0199 & 0.0128 & 0.0118& \textbf{0.0097}& 0.0102\\
       DS3 & 36 & 1812 & 43 &351 & 0.1243 & 0.0882 & 0.0922 &0.0548&\textbf{0.0474} \\
       DS4 & 41 & 1137 & 828& 11505 & 0.0763&0.0637 & 0.0739 & 0.0299&\textbf{0.0267}\\
       DS5 & 50 & 378 &33752& 1516877  &0.8599&0.8218 & 0.8044 & 0.6266&\textbf{0.6199}\\
       DS6 & 50& 1133&35407& 809765&0.3016&0.2786 & 0.2674&0.2360&\textbf{0.2313}\\
       DS7 & 59& 1824&1125&11525 &0.0483&0.0399&0.0301 &0.0191&\textbf{0.0152} \\
       DS8 & 64& 1008&3067&82162 &0.1415&0.1236 &0.1177 & 0.0741&\textbf{0.0563}\\
    \bottomrule
    \end{tabular}
    }
\vskip0.3em
\begin{flushleft}  The ``Sampled trees'' column shows the numbers of unique tree topologies in the training sets. The ``GT trees'' column shows the numbers of unique tree topologies in the ground truth. The results are averaged over 10 replicates. The results of SBN-EM, SBN-EM-$\alpha$ are from \citet{Zhang2018SBN}, and the results of SBN-SGA and ARTree are from \citet{xie2023artree}.
\end{flushleft}
\end{table}

\subsection{Variational Bayesian phylogenetic inference}\label{sec:exp-vbpi}
Our last experiment is on VBPI, where we examine the performance of ARTreeFormer on tree topology posterior approximation.
Following \citet{xie2023artree}, we use the following annealed unnormalized posterior as our target at the $t$-th iteration
\begin{equation}
p(\tau,\bm{q}|\bm{Y},\beta_t)\propto p(\bm{Y}|\tau,\bm{q})^{\beta_t}p(\tau,\bm{q}),
\end{equation}
where $\beta_t=\min\{1,0.001+t/H\}$ is the annealing weight and $H$ is the annealing period.
We use the VIMCO estimator \citep{Mnih2016vimco} and the reparametrization trick \citep{VAE} to obtain the gradient estimates for the tree topology parameters and the branch lengths parameters, respectively.
The results are collected after 400000 parameter updates.

\paragraph{VBPI on DS1-8}
In this part we test the performance of VBPI on the eight standard benchmarks DS1-8, as considered in \citet{Zhang2018SBN, Zhang2019VBPI, Zhang2020VBPINF, Zhang2023VBPIGNN, xie2023artree, xie2024improving, xie2024vbpisibranch}.
We set $H=200000$ for the two more difficult dataset DS6 and DS7, and $H=100000$ for other data sets, following the setting in \citet{xie2023artree}.
We set $K=10$ for the multi-sample lower bound (\ref{eq:lower-bound}).
The results are collected after 400000 parameter updates.
To be fair, for all three VBPI-based methods (VBPI-SBN, VBPI-ARTree, and VBPI-ARTreeFormer), we use the same branch length model that is parametrized by GNNs with edge convolutional operator and learnable topological features as done in \citet{Zhang2023VBPIGNN}.
We also consider two alternative approaches ($\phi$-CSMC \citep{koptagel2022vaiphy}, GeoPhy \citep{mimori2023geophy}) that provide unconfined tree topology distributions and one MCMC based method (MrBayes) as baselines.

The left plot in Fig \ref{fig:ds1-vbpi} shows the lower bound as a function of the number of iterations on DS1. 
We see that although ARTreeFormer converges more slowly than SBN and ARTree at the beginning, it quickly catches up and reaches a similar lower bound in the end.
The middle plot in Fig \ref{fig:ds1-vbpi} shows that both ARTree and ARTreeFormer can provide accurate variational approximations to the ground truth posterior of tree topologies, and both of them outperform SBNs by a large margin.
In the right plot of Fig \ref{fig:ds1-vbpi}, we see that the computation time of ARTreeFormer is substantially reduced compared to ARTree. 
This reduction is especially evident for sampling time since it does not include the branch length generation, likelihood computation, and backpropagation.

\begin{figure}[h]
    \centering
    \includegraphics[width=\linewidth]{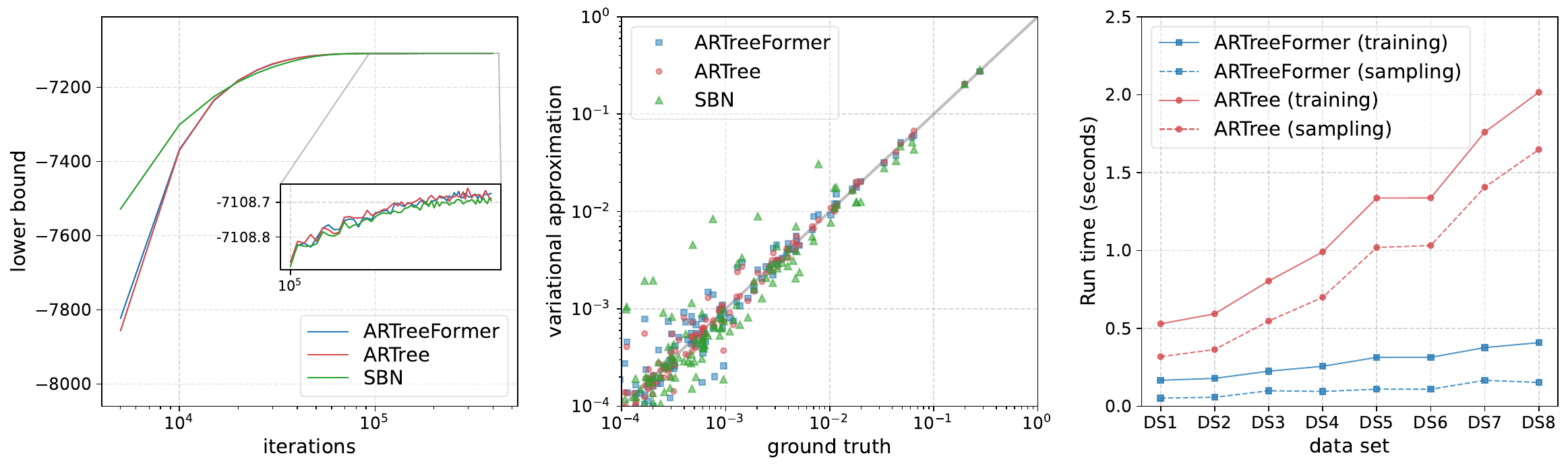}
    \caption{\textbf{Performances of different methods for VBPI.}
    Left: the 10-sample lower bound as a function of the number of iterations on DS1.
    Middle: the variational approximation v.s. the ground truth of the marginal distribution of tree topologies on DS1. 
    Right: Training time per iteration and sampling time (per sampling 10 tree topologies) across different data sets (averaged over 100 trials).
    }
    \label{fig:ds1-vbpi}
\end{figure}

Table \ref{tab:vbpi} shows the marginal likelihood estimates obtained by different methods on DS1-8, including the results of the stepping-stone (SS) method \citep{SS}, which is one of the state-of-the-art sampling based methods for marginal likelihood estimation.
We find that VBPI-ARTreeFormer provides comparable estimates to VBPI-SBN and VBPI-ARTree.
Compared to other VBPI variants, the methodological and computational superiority of ARTreeFormer is mainly reflected by its unconfined support (compared to SBN) and faster computation speed (compared to ARTree).
All VBPI variants perform on par with SS, while the other baselines ($\phi$-CSMC, GeoPhy) tend to provide underestimated results.
We also note that the standard deviations of ARTreeFormer can be smaller than those of ARTree and SBN on most data sets, which can be partially attributed to the potentially more accurate approximation.
Regarding the efficiency-accuracy trade-off, the simplified architecture in ARTreeformer is enough to maintain or even surpass the performance of ARTree.
We also provide more information on the memory and parameter size of different methods for VBPI in Appendix \ref{app:vbpi-results}.
Finally, it is worth noting that VBPI-mixture \citep{vbpimixture, bbvimixture} can provide a better marginal likelihood approximation by employing mixtures of tree models as the variational family.

\begin{table}[h]
\renewcommand\arraystretch{1.2}
\setlength{\tabcolsep}{0.1cm}{}
\caption{\bf Marginal likelihood estimates (in units of nats) of different methods across eight benchmark data sets for Bayesian phylogenetic inference. 
}
\label{tab:vbpi}
\centering
\resizebox{\linewidth}{!}{
\begin{tabular}{lcccccccc}
\toprule
\multicolumn{1}{c}{Data set} & DS1 & DS2 & DS3 & DS4 & DS5 & DS6 & DS7 & DS8 \\
\multicolumn{1}{c}{\# Taxa} & 27 &  29 &  36 & 41 & 50 &  50 & 59   &  64  \\
\multicolumn{1}{c}{\# Sites} & 1949  & 2520 & 1812  & 1137  & 378   & 1133  & 1824  &  1008  \\
\multicolumn{1}{c}{GT trees} & 2784 & 42 & 351 & 11505 & 1516877 & 809765 & 11525 & 82162 \\
\midrule
$\phi$-CSMC~\citep{koptagel2022vaiphy} & -7290.36(7.23)&-30568.49(31.34)&-33798.06(6.62)&-13582.24(35.08)&-8367.51(8.87)&-7013.83(16.99)&N/A &-9209.18(18.03) \\
GeoPhy~\citep{mimori2023geophy}& -7111.55(0.07)&-26368.44(0.13)&-33735.85(0.12)&-13337.42(1.32)&-8233.89(6.63)&-6733.91(0.57)&-37350.77(11.74)&-8660.48(0.78)\\
\rowcolor{gray!20}VBPI-SBN~\citep{Zhang2023VBPIGNN}&-7108.41(0.14)&\textbf{-26367.73(0.07)}&-33735.12(0.09)&-13329.94(0.19)&-8214.64(0.38)&-6724.37(0.40)&\textbf{-37332.04(0.26)}&\textbf{-8650.65(0.45)}\\
\rowcolor{gray!20}VBPI-ARTree~\citep{xie2023artree} & -7108.41(0.19) & \textbf{-26367.71(0.07)} & -33735.09(0.09)&\textbf{-13329.94(0.17)} & -8214.59(0.34) & -6724.37(0.46) &-37331.95(0.27) & -8650.61(0.48)\\
\rowcolor{gray!20}\textbf{VBPI-ARTreeFormer (ours)} &\textbf{-7108.43(0.13)}&\textbf{-26367.71(0.07)}&\textbf{-33735.08(0.08)}&\textbf{-13329.93(0.17)}&\textbf{-8214.63(0.30)}&\textbf{-6724.47(0.35)}& -37331.94(0.31)&-8650.63(0.47)\\
\midrule
MrBayes SS~\citep{SS}&-7108.42(0.18)&-26367.57(0.48)&-33735.44(0.50)&-13330.06(0.54)&-8214.51(0.28)&-6724.07(0.86)&-37332.76(2.42)&-8649.88(1.75)\\
\bottomrule
\end{tabular}
}
\vskip0.3em
\begin{flushleft}
The ``GT trees'' column shows the numbers of unique tree topologies in the ground truth, reflecting the diversity of the phylogenetic posterior.
The marginal likelihood estimates for ARTreeFormer are obtained by importance sampling with 1000 particles from the variational approximation and are averaged over 100 independent runs with standard deviation in the brackets.
A smaller variance is better.
The results of MrBayes SS, which serve as the ground truth, are from \citet{Zhang2019VBPI}.
The results of other methods are reported in their original papers.
\end{flushleft}
\end{table}

\paragraph{VBPI on influenza data} To further test the scalability and vectorization ability of ARTreeFormer, 
we consider the influenza data set with an increasing number - $N=25, 50, 75, 100$ - of nested hemagglutinin (HA) sequences \citep{Zhang22VBPI}.
These sequences were obtained from the NIAID Influenza Research Database (IRD) \citep{zhang2017influenza} through the website at \href{https://www.fludb.org/}{\texttt{https://www.fludb.org/}}, downloading all complete HA sequences that passed quality control, which were then subset to H7 sequences, and further downsampled using the Average Distance to the Closest Leaf (ADCL) criterion \citep{matsen2013minimizing}.
For all the VBPI based methods - SBN, ARTree, and ARTreeFormer, we set $H=100000$, and the results are collected after 400000 parameter updates.
For ARTree and ARTreeFormer, we use the same branch length model that is parametrized by GNNs with edge convolutional operator and learnable topological features as done in \citet{Zhang2023VBPIGNN}; for SBN, we use the bipartition-feature-based branch length model considered in \citet{Zhang2019VBPI}.

Table \ref{tab:influenza-mll} reports the marginal likelihood estimates of different methods on the influenza data set. We see that all three VBPI methods yield very similar marginal likelihood estimates to SS when $N=25, 50$.
For a larger number of sequences $N=75, 100$, SS tends to provide higher marginal
likelihood estimates than VBPI methods, albeit with larger variances which indicates the decreasing reliability of those estimates.
On the other hand, the variances of the estimates provided by VBPI methods are much smaller which implies more reliable estimates \citep{Zhang22VBPI}.
Compared to ARTree, ARTreeFormer can provide much better MLL estimates (also closer to SBN) while maintaining a relatively small variance, striking a better balance between approximation accuracy and reliability.

\begin{table}[h]
    \centering
\caption{\bf The marginal likelihood estimates (in units of nats) of different methods on the influenza data with up to 100 taxa.}
\resizebox{0.8\linewidth}{!}{
    \begin{tabular}{ccccc}
    \toprule
 Subset size ($N$) & MrBayes SS & VBPI-SBN & VBPI-ARTree& VBPI-ARTreeFormer  \\
    \midrule
25 & -13378.23(0.24) & -13378.38(0.06)&-13378.39(0.06)&-13378.38(0.06)\\ 
50 & -18615.82(1.57) &-18615.40(0.16) & -18615.31(0.18)&-18615.30(0.20)\\ 
75 &-23647.14(13.25) &-23681.85(0.27) & -23849.85(0.30)&-23763.58(0.29)\\ 
    100 &  -28176.80(47.16) & -28556.96(0.36)&-29416.42(0.44) & -28650.72(0.32)\\
    \bottomrule
    \end{tabular}
}
\begin{flushleft}
The results of MrBayes SS and VBPI-SBN are reported by \citet{Zhang22VBPI}, and those of VBPI-ARTree and VBPI-ARTreeFormer are produced by us.
\end{flushleft}
    \label{tab:influenza-mll}
\end{table}

\section{Discussion}

\paragraph{Comparison with prior works}
The most common approach for Bayesian phylogenetic inference is Markov chain Monte Carlo (MCMC), which relies on random walks to explore the tree space, e.g., MrBayes \citep{ronquist2012mrbayes}, BEAST \citep{drummond2007beast}. 
MCMC methods have long been considered the standard practice of systematic biology research and are used to construct the ground truth phylogenetic trees in our experiments.
However, as the tree space contains both the continuous and discrete components (i.e., the branch lengths and tree topologies), the posterior distributions of phylogenetic trees are often complex multimodal distributions.
Furthermore, the involved tree proposals are often limited to local modifications that can lead to low exploration efficiency, which makes 
MCMC methods require extremely long runs to deliver accurate posterior estimates \citep{Whidden2014QuantifyingME,Zhang22VBPI}.

ARTreeFormer is established in the line of variational inference (VI) \citep{jordan1999introduction, Blei2016VariationalIA}, another powerful tool for Bayesian inference.
VI selects the closest member to the posterior distribution from a family of candidate variational distributions by minimizing some statistical distance, usually the KL divergence.
Compared to MCMC, VI tends to be faster and easier to scale up to large data by transforming a sampling problem into an optimization problem.
The success of VI often relies on the design of expressive variational families and efficient optimization procedures.
Besides the variational Bayesian phylogenetic inference (VBPI) introduced before, there exist other VI methods for Bayesian phylogenetic inference.
VaiPhy \citep{koptagel2022vaiphy} approximates the posterior of multifurcating trees with a novel sequential tree topology sampler based on maximum spanning trees.
GeoPhy \citep{mimori2023geophy} models the tree topology distribution through a mapping from continuous distributions over the leaf nodes to tree topologies via the Neighbor-Joining (NJ) algorithm \citep{NJ}.
PhyloGen \citep{duan2024phylogen} uses pre-trained DNA-based node features for computing the pairwise distance matrix, which will then be mapped to a binary tree topology with the NJ algorithm.

As a classical tool in Bayesian statistics, sequential Monte Carlo (SMC) \citep{BouchardCt2012PhylogeneticIV} and its variant combinatorial SMC (CSMC) \citep{Wang2015BayesianPI} propose to sample tree topologies through subtree merging and resampling steps for Bayesian phylogenetic inference. 
VCSMC \citep{Moretti2021VariationalCS} employs a learnable proposal distribution based on CSMC and optimizes it within a variational framework.
$\phi$-CSMC \citep{koptagel2022vaiphy} makes use of the parameters of VaiPhy to design the proposal distribution for sampling bifurcating trees.
This approach is further developed by H-VCSMC \citep{chen2025hvcsmc} which transfers the merging and resampling steps of VCSMC to the hyperbolic space.
The subtree merging operation in SMC based methods is also the core idea of PhyloGFN \citep{Zhou2023PhyloGFN}, which instead treats the merging choices as actions within the GFlowNet \citep{bengio2021gflownet} framework and optimizes the trajectory balance objective \citep{malkin2022trajectorybalance}.

\paragraph{Potential impact of ARTreeFormer}
Building upon the tree topology construction algorithm of ARTree, ARTreeFormer introduces a more computationally efficient and expressive distribution family for variational Bayesian phylogenetic inference (VBPI). The efficiency gains primarily stem from the use of a fixed-point algorithm in the node embedding module. 
While fixed-point algorithms can often raise concerns regarding the cost of matrix multiplications and potentially long convergence times—especially when poorly tuned—ARTreeFormer addresses these challenges effectively in several ways.
First, matrix multiplications are implemented as tensor operations, which are efficiently accelerated on CUDA-enabled devices (see our open-source implementation). Second, the number of iterations required for convergence is significantly reduced through the use of the power trick, achieving logarithmic scaling. 
Thirdly, we provide a theoretical guarantee (Corollary \ref{thm:convergence}) that the convergence rate of the fixed-point algorithm is constant, independent of the number of taxa or the shape of the tree topology.

Topological node embeddings (i.e., learnable topological features) \citep{Zhang2023VBPIGNN} provide a general-purpose representation framework for phylogenetic trees and have been employed in various downstream tasks.
For example, VBPI-SIBranch \citep{xie2024vbpisibranch} uses these embeddings to parametrize semi-implicit branch length distributions, while PhyloVAE \citep{xie2025phylovae} leverages them to obtain low-dimensional representations of tree topologies for tree clustering and diagnostic analysis in phylogenetics. The fixed-point algorithm introduced in this work offers an improved and efficient approach to computing these embeddings, and can be seamlessly integrated into such downstream applications, demonstrating broad potential for impact across phylogenetic modeling tasks.

Another key contribution of ARTreeFormer is the integration of the attention mechanism \citep{vaswani2017attention} into phylogenetic inference.
Since its introduction, attention has become a foundational component in modern deep learning, powering numerous milestone models such as GPT-4o \citep{openai2024gpt4o} and DeepSeek-V2 \citep{deepseek2024}.
Despite its widespread success, its potential for modeling phylogenetic tree structures remains underexplored. In this work, we demonstrate that incorporating attention into the message passing module of ARTreeFormer enables comparable or superior performance relative to traditional graph neural networks (GNNs), highlighting its effectiveness in capturing long-range dependencies in tree-structured data.

Phylogenetic inference provides critical insights for making informed public health decisions, particularly during pandemics. Developing efficient Bayesian phylogenetic inference algorithms that can deliver accurate posterior estimates in a timely manner is therefore of immense value, with the potential to save countless lives.
VI approaches hold significant promise due to their optimization-based framework.
For example, VI methods have been used for rapid analysis of pandemic-scale data (e.g.,  SARS-CoV-2 genomes) to provide accurate estimates of epidemiologically relevant quantities that can be corroborated via alternative public health data sources \citep{Ki2022-nt}.
We expect more efficient VI approaches for Bayesian phylogenetics and associated software to be developed in the near future, further advancing this critical field.

\paragraph{Future directions}
There are several future practical directions for advancing ARTreeFormer, which we discuss as follows.
Firstly, the embedding method for phylogenetic trees in ARTreeFormer can be further explored. 
For example, PhyloGen \citep{duan2024phylogen} use pre-trained DNA-based node features, and GeoPhy \citep{mimori2023geophy} and H-VCSMC \citep{chen2025hvcsmc} consider embedding trees in hyperbolic space.
As the input to the model, the representation power and generalization ability of the embedding method might have a marked impact on the performance of ARTreeFormer.
Secondly, the attention mechanism for the message passing on phylogenetic trees can be more delicately designed.
For example, the attention masks can be modified according to the neighborhood structures.
\citet{muller2024attending} provides a comprehensive survey on the design details of graph transformers.
Thirdly, the fast computation and scalability of ARTreeFormer offer the possibility of large phylogenetic inference models capable of zero-shot inference on biological sequences.
This may require more expressive model designs, especially powerful node embedding schemes, and more high-quality data.
We hope these discussions could help inspire more advances in variational approaches for phylogenetic inference.

\section{Conclusion}
In this work, we presented ARTreeFormer, a variant of ARTree that leverages the scalable fixed-point iteration algorithm and the attention mechanism to accelerate the autoregressive modeling of tree topologies in phylogenetic inference.
In contrast to ARTree, which involves the Dirichlet energy minimization via expensive and non-vectorizable tree traversals to compute the node embeddings, ARTreeFormer introduce a specially designed fixed-point algorithm that facilitates highly vectorizable computation.
We also introduce an attention-based global message passing module, which is capable of capturing the crucial global information in only one forward pass, to replace the GNN-based local message passing module.
Experiments on various phylogenetic inference problems showed that ARTreeFormer is significantly faster than ARTree in training and evaluation while performing comparably or better in terms of approximation accuracy.

\section*{Acknowledgements}
Cheng Zhang was partially supported by National Natural Science Foundation of China (grant no. 12201014, grant no. 12292980 and grant no. 12292983), as well as National Institutes of Health grant AI162611.
The research of Cheng Zhang was support in part by National Engineering Laboratory for Big Data Analysis and Applications, the Key Laboratory of Mathematics and Its Applications (LMAM) and the Key Laboratory of Mathematical Economics and Quantitative Finance (LMEQF) of Peking University.
The authors appreciate Zichao Yan, Ming Yang Zhou, and Dinghuai Zhang for their constructive discussion on this project. 
The authors are grateful for the computational resources provided by the High-performance Computing Platform of Peking University.

\paragraph{Data availability statement}
The sequence data for the DS1-8 can be found at \\ \href{https://github.com/tyuxie/ARTreeFormer}{\texttt{https://github.com/tyuxie/ARTreeFormer}}.
For reproducing the TDE task, the short run data for DS1-4 are provided in \href{https://github.com/tyuxie/ARTreeFormer}{\texttt{https://github.com/tyuxie/ARTreeFormer}}, and those for DS5-8 are provided at \href{https://drive.google.com/drive/folders/1qMdv_NxpsLZlu510izs26V6b02smGAoH}{\texttt{https://drive.google.com/drive/folders/1qMdv\_NxpsLZlu510izs26V6b02smGAoH}}.

\paragraph{Code availability statement}
The codebase for reproducing the results of ARTreeFormer is provided at \href{https://github.com/tyuxie/ARTreeFormer}{\texttt{https://github.com/tyuxie/ARTreeFormer}}.

\bibliographystyle{nameyear}
\bibliography{reference}

\newpage
\appendix

\section{Details of ARTree}\label{app:artree}

\subsection{Tree topology generating process}
Let $\tau_n=(V_n,E_n)$ be a tree topology with $n$ leaf nodes and $V_n, E_n$ are the sets of nodes and edges respectively. 
Here we only discuss the modeling of unrooted tree topologies.
A pre-selected order (also called the taxa order) for the leaf nodes $\mathcal{X}=\{x_1,\ldots,x_N\}$ is assumed. We first give the definition of ordinal tree topologies.

\begin{definition}[Ordinal Tree Topology; Definition 1 in \citet{xie2023artree}]
Let $\mathcal{X}=\{x_1,\ldots,x_N\}$ be a set of $N(N\geq 3)$ leaf nodes. 
Let $\tau_n=(V_n,E_n)$ be a tree topology with $n (n\leq N)$ leaf nodes in $\mathcal{X}$. 
We say $\tau_n$ is an ordinal tree topology of rank $n$, if its leaf nodes are the first $n$ elements of $\mathcal{X}$, i.e., $V_n\cap \mathcal{X}= \{x_1,\ldots,x_n\}$. 
\end{definition}

The tree topology generating process is initialized by $\tau_3$, the unique ordinal tree topology of rank 3.
In the $n$-th step ($n$ start from 3), assume we have an ordinal tree topology $\tau_n = (V_n, E_n)$ of rank $n$. To incorporate the leaf node $x_{n+1}$ into $\tau_n$, the following steps are taken:
\begin{enumerate}\itemsep0em 
    \item A choice is made for an edge $e_{n} = (u,v) \in E_n$, which is then removed from $E_n$.
    \item Add a new node $w$ and two additional edges, $(u,w)$ and $(w,v)$ to the tree topology $\tau_n$.
    \item Add the next leaf node $x_{n+1}$ and an additional edge $(w,x_{n+1})$ to the tree topology $\tau_n$.
\end{enumerate}
The above steps create an ordinal tree topology $\tau_{n+1}$ of rank $n+1$.
Repeating these steps for $n=3,\ldots,N-1$ leads to the eventual formation of the ordinal tree topology $\tau = \tau_N$ of rank $N$.
The selected edges at each time step form a sequence $D = (e_3, \ldots, e_{N-1})$, which we call $D$ a decision sequence. Here we give two main theoretical results.
\begin{theorem}\label{thm-bijection}
The generating process $g(\cdot):D\mapsto\tau$ is a bijection between the set of decision sequences of length $N-3$ and the set of ordinal tree topologies of rank $N$.
\end{theorem}
\begin{theorem}
The time complexity of the decomposition process induced by $g^{-1}(\cdot)$ is $O(N)$.
\end{theorem}
The bijectiveness in Theorem \ref{thm-bijection} implies that we can model the distribution $Q(\tau)$ over tree topologies by modelling $Q(D)$ over decision sequences, i.e.,
\begin{equation}\label{eq-prob-decomp}
Q(\tau) = Q(D) = \prod_{n=3}^{N-1}Q(e_n|e_{<n}),
\end{equation}
where $e_{<n}=(e_3,\ldots,e_{n-1})$ and $e_{<3}=\emptyset$. 
The conditional distribution $Q(e_n|e_{<n})$, which describes the distribution of edge decision given all the decisions made previously, is called the edge decision distribution by us.

\subsection{Graph neural networks for edge decision distribution}
The edge decision distribution $Q(e_n|e_{<n})$ defines the probability of adding the leaf node $x_{n+1}$ to the edge $e_n$ of $\tau_n$, conditioned on all the ordinal tree topologies $(\tau_3,\ldots,\tau_{n})$ generated so far. 
To model $Q(e_n|e_{<n})$, ARTree employs the following four modules.
\paragraph{Node embedding module}
At the $n$-th step of the generation process, ARTree relies on the node embedding module to assign node embeddings for the nodes of the current tree topology $\tau_n = (V_n, E_n)$. 
The embedding method follows \citet{Zhang2023VBPIGNN}, which first assigns one-hot encoding for the leaf nodes:
\[
\left[f_n(x_i)\right]_j = \delta_{ij}, \quad 1 \leq i \leq n, \quad 1 \leq j \leq N,
\]
where $\delta$ denotes the Kronecker delta function. We then obtain embeddings for the interior nodes by minimizing the Dirichlet energy, defined as 
\[
\ell(f_n, \tau_n) := \sum_{(u,v)\in E_n}||f_n(u)-f_n(v)||^2.
\]
This minimization process is achieved through the two-pass algorithm (Algorithm \ref{alg:embedding}).
Note that this process contains $(2n-6)$ sub-iterations and each sub-iteration contains a linear combination over at most 3 vectors in $\mathbb{R}^N$.
The time complexity of calculating the topological node embeddings is $O(Nn)$.
Finally, a linear transformation is applied to all the node embeddings to obtain the initial node features in $\mathbb{R}^d$ for message passing. 
It should be highlighted that the embeddings for interior nodes may vary as the number of leaf nodes $n$, leading to the need for time guidance in the readout module.

\begin{algorithm}[h]
\caption{ARTree: an autoregressive model for phylogenetic tree topologies \citep{xie2023artree}}
\label{alg:generation}
\KwIn{A set $\mathcal{X}=\{x_1,\ldots,x_N\}$ of leaf nodes.}
\KwOut{An ordinal tree topology $\tau$ of rank $N$; the ARTree probability $Q(\tau)$ of $\tau$.}
$\tau_3=(V_3,E_3) \leftarrow$ the unique ordinal tree topology of rank $3$\;
   \For{$n=3,\ldots,N-1$}{
Let $f_n(u) = c_u f_n(\pi_u)+d_u$ where $\pi_u$ is the parent of $u$\;
Calculate the probability vector $q_n\in \mathbb{R}^{|E_n|}$ using the current GNN model\;
Sample an edge decision $e_n$ from $\mathrm{\sc Discrete}\left(q_n\right)$ and assume $e_{n}=(u,v)$\;
Create a new node $w$\;
$E_{n+1} \leftarrow \left(E_n\backslash \{e_{n}\}\right)\cup \{(u,w), (w,v), (w,x_{n+1})\}$\;
$V_{n+1} \leftarrow V_n\cup \{w,x_{n+1}\}$\;
$\tau_{n+1}\leftarrow (V_{n+1}, E_{n+1})$\;
}
 $\tau\leftarrow \tau_N$\;
 $Q(\tau)\leftarrow q_3(e_3)q_4(e_4)\cdots q_{N-1}(e_{N-1}).$
\end{algorithm}

\paragraph{Message passing module}
ARTree employs iterative message passing rounds to calculate the node features, capturing the topological information of $\tau_n$. 
The $l$-th message passing round is implemented by 
\begin{align*}
m^l_n(u,v) &=F_{\textrm{message}}^l(f^l_n(u), f^l_n(v)),\\
f^{l+1}_n(v) &= F_{\textrm{updating}}^l\left(\{m^l_n(u,v);u\in \mathcal{N}(v)\}\right),
\end{align*}
where $F_{\textrm{message}}^l$ and $F_{\textrm{updating}}^l$ are the message function and updating function in the $l$-th round, and $\mathcal{N}(v)$ is the neighborhood of the node $v$.
The corresponding time-complexity is $O(nd^2)$ (noting that MLPs are applied to all the nodes)
In particular, ARTree sets the number of message passing steps $L=2$ and utilizes the edge convolution operator~\citep{Wang2018DynamicGC} for the design of $F_{\textrm{message}}^l$ and $F_{\textrm{updating}}^l$.

\paragraph{Recurrent module}
To efficiently incorporate the information of previously generated tree topologies into the edge decision distribution, ARTree uses a gated recurrent unit (GRU) \citep{cho2014learning} to form the hidden states of each node.
Concretely, the recurrent module is implemented by
\[
h_n(v) = \mathrm{\sc GRU}(h_{n-1}(v), f_n^L(v)),
\]
where $h_n(v)$ is the hidden state of $v$ at the $n$-th step in the generating process.
For the newly added nodes, their hidden states are initialized to zeros.
This module is mainly composed of MLPs on the node/edge features, whose time complexity is $O(nd^2)$. 

\paragraph{Readout module}
In the readout module, to form the edge decision distribution $Q(e_n|e_{<n})$, ARTree calculates the scalar edge feature $r_n(e)\in\mathbb{R}$ of $e=(u,v)$ using
\begin{align*}
p_n(e) &= F_{\textrm{pooling}}\left(h_n(u)+b_n, h_n(v)+b_n\right), \\
r_{n}(e) &= F_{\textrm{readout}}\left(p_n(e)+b_n\right),
\end{align*}
where $b_n$ is the sinusoidal positional embedding of time step $n$ that is widely used in Transformers \citep{vaswani2017attention},
$F_{\textrm{pooling}}$ is the pooling function implemented as 2-layer MLPs followed by an elementwise maximum operator, and $F_{\textrm{readout}}$ is the readout function implemented as 2-layer MLPs with a scalar output. 
This module is mainly composed of MLPs on the node/edge features, whose time complexity is $O(nd^2)$. 
The edge decision distribution is
\[
Q(\cdot|e_{<n})\sim \mathrm{\sc Discrete}\left(q_n\right),\quad q_n  = \mathrm{softmax}\left(\{r_n(e)\}_{e\in E_n}\right),
\]
where $q_n\in \mathbb{R}^{|E_n|}$ is a probability vector.

Let $\bm{\phi}$ be all the learnable parameters in GNNs. Then the ARTree based probability of a tree topology $\tau$ takes the form
\[
Q_{\bm{\phi}}(\tau) = Q_{\bm{\phi}}(D) = \prod_{n=3}^{N-1}Q_{\bm{\phi}}(e_n|e_{<n}),
\]
The whole process of ARTree for generating a tree topology is summarized in Algorithm \ref{alg:generation}.

\begin{algorithm}[h]
\caption{Two-pass algorithm for topological embeddings for internal nodes \citep{Zhang2023VBPIGNN}}
\label{alg:embedding}
\KwIn{Tree topology $\tau_n=(V_n,E_n)$ of rank $n$, where $V_n=V_n^b\cup V_n^o$; Topological embeddings for the leaf nodes $\{f_n(u)|u\in V_n^b\}$.}
\KwOut{Topological embeddings for the leaf nodes $\{f_n(u)|u\in V_n^o\}$}
Initialized $c_u=0, d_u=f_n(u)|u\in V_n^b$\;
   \For{$u$ in the postorder traverse of $\tau_n$}{
   \uIf{$u$ is not the root node}{
   Compute 
   \[
   c_u = \frac{1}{|\mathcal{N}(u)|-\sum_{v\in \mathrm{ch}(u)}c_v},\quad d_u = \frac{\sum_{v\in \mathrm{ch}(u)}d_v}{|\mathcal{N}(u)|-\sum_{v\in \mathrm{ch}(u)}c_v}
   \]
   where $\mathcal{N}(u)$ is the neighborhood of $u$ and $\mathrm{ch}(u)$ is the set of the children of $u$.
}
}
   \For{$u$ in the preorder traverse of $\tau_n$}{
   \uIf{$u$ is not the root node}{
    Let $f_n(u) = c_u f_n(\pi_u)+d_u$ where $\pi_u$ is the parent of $u$.
    }
   \Else{Let $f_n(u) = \frac{\sum_{v\in \mathrm{ch}(u)}d_v}{|\mathcal{N}(u)|-\sum_{v\in \mathrm{ch}(u)}c_v}$.
   }
   }
\end{algorithm}

\section{Details of variational Bayesian phylogenetic inference}\label{app:vbpi}
By positing a tree topology variational distribution $Q_{\bm{\phi}}(\tau)$ and a branch length variational distribution $Q_{\bm{\psi}}(\bm{q}|\tau)$ which is conditioned on tree topologies, the variational Bayesian phylogenetic inference (VBPI) \citep{Zhang2019VBPI} approximates the phylogenetic posterior $p(\tau, \bm{q}|\bm{Y})$ in Eq~(\ref{eq:posterior}) with  $Q_{\bm{\phi}, \bm{\psi}}(\tau,\bm{q})=Q_{\bm{\phi}}(\tau)Q_{\bm{\psi}}(\bm{q}|\tau)$.
To find the best approximation, VBPI maximizes the following multi-sample lower bound
\[
L^{K}(\bm{\phi},\bm{\psi}) = \mathbb{E}_{Q_{\bm{\phi},\bm{\psi}}(\tau^{1:K},\bm{q}^{1:K})}\log \left(\frac{1}{K}\sum_{i=1}^K\frac{p(\bm{Y}|\tau^i,\bm{q}^i) p(\tau^i, \bm{q}^i)}{Q_{\bm{\phi}}(\tau^i)Q_{\bm{\psi}}(\bm{q}^i|\tau^i)}\right).
\]
where $Q_{\bm{\phi},\bm{\psi}}(\tau^{1:K},\bm{q}^{1:K})=\prod_{i=1}^K Q_{\bm{\phi},\bm{\psi}}(\tau^{i},\bm{q}^{i})$.
Compared to the single-sample lower bound, the multi-sample lower bound enables efficient variance-reduced gradient estimators and encourages exploration over the vast and multimodal tree space.
However, as a large $K$ may also reduce the signal-to-noise ratio and deteriorate the training of variational parameters \citep{Rainforth19}, a moderate $K$ is suggested \citep{Zhang22VBPI}.
In practice, the gradients of the multi-sample lower bound w.r.t the tree topology parameters $\bm{\phi}$ and the branch length parameter $\bm{\psi}$ can be estimated by the VIMCO/RWS estimator \citep{Mnih2016vimco,RWS} and the reparameterization trick \citep{VAE} respectively.
Specifically, the gradient $\nabla_{\bm{\phi}} L^K(\bm{\phi},\bm{\psi})$
 can be expressed as 
 \begin{align*}
 \nabla_{\bm{\phi}} L^K(\bm{\phi},\bm{\psi}) &= R_1+R_2,\\
R_1 &= \mathbb{E}_{Q_{\bm{\phi},\bm{\psi}}(\tau^{1:K},\bm{q}^{1:K})}\nabla_{\bm{\phi}}\log \left(\frac{1}{K}\sum_{i=1}^K\frac{p(\bm{Y}|\tau^i,\bm{q}^i) p(\tau^i, \bm{q}^i)}{Q_{\bm{\phi}}(\tau^i)Q_{\bm{\psi}}(\bm{q}^i|\tau^i)}\right) \\ 
R_2 &= \mathbb{E}_{Q_{\bm{\phi},\bm{\psi}}(\tau^{1:K},\bm{q}^{1:K})}\sum_{i=1}^K\log \left(\frac{1}{K}\sum_{i=1}^K\frac{p(\bm{Y}|\tau^i,\bm{q}^i) p(\tau^i, \bm{q}^i)}{Q_{\bm{\phi}}(\tau^i)Q_{\bm{\psi}}(\bm{q}^i|\tau^i)}\right)\nabla_{\bm{\phi}}Q_{\bm{\phi},\bm{\psi}}(\tau^{i},\bm{q}^{i}).
 \end{align*}
VIMCO considers the following expression of $R_2$,
\[
R_2 = \mathbb{E}_{Q_{\bm{\phi},\bm{\psi}}(\tau^{1:K},\bm{q}^{1:K})}\sum_{i=1}^K\left\{\log \left(\frac{1}{K}\sum_{i=1}^K\frac{p(\bm{Y}|\tau^i,\bm{q}^i) p(\tau^i, \bm{q}^i)}{Q_{\bm{\phi}}(\tau^i)Q_{\bm{\psi}}(\bm{q}^i|\tau^i)}\right)-\hat{f}_i\right\}\nabla_{\bm{\phi}}Q_{\bm{\phi},\bm{\psi}}(\tau^{i},\bm{q}^{i})
\]
where $\hat{f}_i=\log\left(\frac{1}{K-1}\sum_{j\neq i}\frac{p(\bm{Y}|\tau^j,\bm{q}^j) p(\tau^j, \bm{q}^j)}{Q_{\bm{\phi}}(\tau^j)Q_{\bm{\psi}}(\bm{q}^j|\tau^j)}\right)$ is a control variate.

The tree topology model $Q_{\bm{\phi}}(\tau)$ can be parametrized by ARTree, which enjoys unconfined support over the tree topology space.
In addition to ARTree, subsplit Bayesian networks (SBNs) have long been the common choice for $Q_{\bm{\phi}}(\tau)$. 
In SBNs, a subset $C$ of the leaf nodes is called a clade, and an ordered pair of two clades $(C_1, C_2)$ is called a subsplit of $C$ if $C_1\cup C_2=C$.
For each internal node on a tree topology $\tau$, it corresponds to a subsplit $s$ determined by the descendant leaf nodes of its children.
The SBNs are then parametrized by the probabilities of the root subsplit $\{p_{s_1};s_1\in\mathbb{S}_{\mathrm{r}}\}$ and the probabilities of the child-parent subsplit pairs $\{p_{s|t}; \ s|t\in \mathbb{S}_{\mathrm{ch|pa}}\}$.
For an unrooted tree topology $\tau=(V,E)$, its SBN based probability is 
\[
Q_{\mathrm{sbn}}(\tau) = p_{s_r}\prod_{u\in V^o;u\neq r}p_{s_u|s_{\pi_u}},
\]
where $V^o$ is the set of internal nodes, $r$ is the root node, $\pi_u$ are the parents of $u$, and $s_{u}$ is the subsplit assignment of the node $u$.
As the size of $\mathbb{S}_{\mathrm{r}}$ and $\mathbb{S}_{\mathrm{ch|pa}}$ explodes combinatorially as the number of taxa increases, SBNs rely on subsplit support estimation for a tractable parameterization.
\textbf{The subsplit support estimation can be difficult when the phylogenetic posterior is diffuse, and makes the support of SBNs cannot span the entire tree topology space.}
We refer the readers to \citet{Zhang2018SBN} and \citet{Zhang2019VBPI} for a detailed introduction to SBNs as well as their application to VBPI.

The branch length model $Q_{\bm{\psi}}(\bm{q}|\tau)$ is often taken to be a diagonal lognormal distribution, which can be parametrized using the learnable topological features \citep{Zhang2023VBPIGNN} of $\tau$ as follows.
This approach first assigns the topological node embeddings $\{f_u\}_{u\in V}$ to the nodes on $\tau$ (Algorithm \ref{alg:embedding}) and then forms the node features $\{h_u\}_{u\in V}$ using message passing networks over $\tau$.
Usually, these message passing networks take the edge convolutional operator \citep{Wang2018DynamicGC}.
For each edge $e=(u,v)$ in $\tau$, one can obtain the edge features using $h_e=p(h_u,h_v)$ where $p$ is a permutation invariant function called the edge pooling. At last, the mean and standard deviation parameters for the diagonal lognormal distribution are given by 
\begin{equation*}
\mu(e,\tau) = \mathrm{MLP}^{\mu}(h_e),\quad \sigma(e,\tau)=\mathrm{MLP}^{\sigma}(h_e)
\end{equation*}
where $\mathrm{MLP}^{\mu}$ and $\mathrm{MLP}^{\sigma}$ are two multi-layer perceptrons (MLPs).
In the VBPI experiment in Section \ref{sec:exp-vbpi}, the collaborative branch length models for all SBN, ARTree, and ARTreeFormer are parametrized in this way.

\section{Additional experimental results}\label{app:results}
\subsection{Additional results on tree topology density estimation}\label{app:tde-results}

Fig \ref{fig:tde} shows the performance of different methods on DS1. 
Both ARTree and ARTreeFormer provide more accurate probability estimates for the tree topologies on the two peaks of the posterior distribution, compared to SBN-EM and SBN-SGA.
We see that ARTreeFormer can provide the same accurate probability estimates as ARTree, which proves the effectiveness of ARTreeFormer.

\begin{figure}[h]
    \centering
    \includegraphics[width=0.95\linewidth]{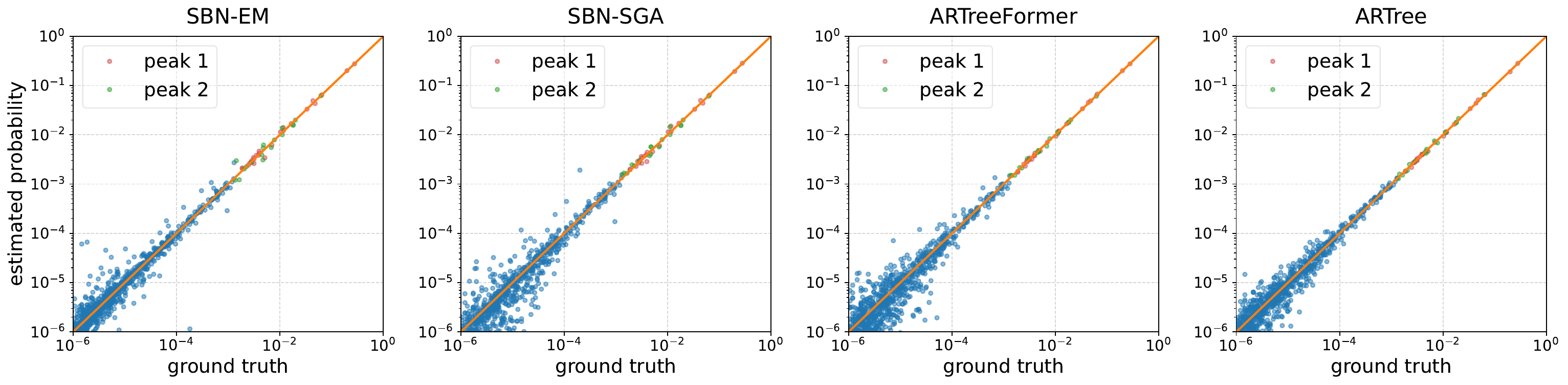}
    \caption{\textbf{Performances of different methods for tree topology density estimation on DS1.}}
    \label{fig:tde}
\end{figure}

For ARTreeFormer, we also conducted an ablation study about the number of heads $h$ and the hidden dimension $d$ in the multi-head attention block (Table \ref{tab:tde-ablation}).
For DS1-4, we train the ARTreeFormer model on the ground truth data set with a batch size of 10 and a learning rate of 0.0001, and evaluate the KL divergence towards the ground truth after 200,000 iterations.
In most cases, the KL divergence gets better as the hidden dimension $d$ increases, while it is not so sensitive to the number of heads.

\begin{table}[h]
\renewcommand\arraystretch{1.2}
\setlength{\tabcolsep}{0.1cm}{}
\caption{\textbf{KL divergences ($\downarrow$) to the ground truth obtained by ARTreeFormer with different hyper-parameters on TDE.}
}
\label{tab:tde-ablation}
\centering
\vskip0.3em
\begin{tabular}{lcccc}
\toprule
\multicolumn{1}{c}{Hyper-parameters} & $h=2, d=100$&$h=4, d=100$&$h=4, d=200$&$h=8, d=200$ \\
\midrule
DS1&0.0058&0.0060&0.0039&0.0039\\
DS2&0.0002&0.0002&0.0003&0.0002\\
DS3&0.0058&0.0052&0.0055&0.0054\\
DS4&0.0097&0.0101&0.0069&0.0071\\
\bottomrule
\end{tabular}
\end{table}

\subsection{Additional results on variational Bayesian phylogenetic inference}\label{app:vbpi-results}

To fully demonstrate the computational burden of ARTreeFormer compared to ARTree, we report the parameter size and memory usage of ARTreeFormer and ARTree for VBPI in Table \ref{tab:vbpi-parameter-memory}.
We see that ARTreeFormer has less memory consumption compared to ARTree, because ARTreeFormer does not need to update all the node features on the tree topology, in analogy with the shorter sequence length in language modeling.

\begin{table}[h]
\renewcommand\arraystretch{1.2}
\setlength{\tabcolsep}{0.1cm}{}
\caption{\textbf{The parameter size and memory usage of ARTreeFormer and ARTree for VBPI.}
}
\label{tab:vbpi-parameter-memory}
\centering
\vskip0.3em
\resizebox{\linewidth}{!}{
\begin{tabular}{lcccccccc}
\toprule
\multicolumn{1}{c}{Data set} & DS1 & DS2 & DS3 & DS4 & DS5 & DS6 & DS7 & DS8 \\
\midrule
ARTree (learnable parameter size)&194K&195K&197K&199K&203K&203K&207K&209K\\
ARTreeFormer (learnable parameter size)&215K&216K&216K&217K&218K&218K&219K&219K\\
\midrule
ARTree (memory)&1143MB&1395MB&1376MB&1680MB&1817MB&1698MB&2070MB&2148MB\\
ARTreeFormer (memory)&556MB&577MB&630MB&690MB&798MB&794MB&896MB&1044MB\\
\bottomrule
\end{tabular}
}
\end{table}

\end{document}